\documentclass[a4paper,english]{elsarticle}

\usepackage[utf8]{inputenc}
\usepackage{graphicx}
\usepackage{xspace}
\usepackage{tikz}
\usetikzlibrary{fit,calc}
\usepackage{times}
\usepackage{float,verbatim,amssymb,amsmath}
\usepackage{subcaption}
\usepackage{enumitem}
\newtheorem{definition}{Definition}
\newtheorem{lemma}{Lemma}
\newtheorem{observation}{Observation}

\newtheorem{remark}{Remark}
\newtheorem{theorem}{Theorem}

\newtheorem{labrule}{Labeling rule}
\usepackage{placeins}
\usepackage{varioref}
\usepackage{setspace}
\graphicspath{{./figs/}}
\tikzset{defnode/.style={draw, circle, fill=gray, inner sep=.6pt}}
\tikzset{mis/.style={draw, circle, red, inner sep=1.5pt}}
\xdefinecolor{red}{rgb}{0.8,0.1,0.1}
\xdefinecolor{green}{rgb}{0.0,0.7,0.0}

\floatstyle{ruled}
\newfloat{algorithm}{tp}{lop}
\floatname{algorithm}{Algorithm}

\newcommand{\forallMIS}{\ensuremath{{\cal RMIS^\forall}}\xspace}
\newcommand{\existsMIS}{\ensuremath{{\cal RMIS^\exists}}\xspace}

\newcommand{\ie}{{i.e.,}\xspace}
\newcommand{\eg}{{e.g.}\xspace}

\newcommand{\typePI}{\textsl{PI}\xspace}
\newcommand{\typePIdesc}{Possibly In\xspace}

\newcommand{\typePO}{\textsl{PO}\xspace}
\newcommand{\typePOdesc}{Possibly Out\xspace}

\newcommand{\typePE}{\textsl{PE}\xspace}
\newcommand{\typePEdesc}{Possibly External\xspace}

\newcommand{\ABCT}{\ensuremath{\mathcal{ABC}}-tree\xspace}

\newcommand{\A}{\ensuremath{\mathcal{A}}\xspace}
\newcommand{\B}{\ensuremath{\mathcal{B}}\xspace}
\newcommand{\C}{\ensuremath{\mathcal{C}}\xspace}
\newcommand{\PV}{\ensuremath{\mathcal{P}}\xspace}

\newcommand{\children}{{\tt children}\xspace}
\newcommand{\parent}{{\tt parent}\xspace}
\newcommand{\descendants}{{\tt descendants}\xspace}

\newcommand{\R}{\ensuremath{R}\xspace}

\newenvironment{proof}{\noindent\textbf{Proof:}}{{\hfill $\Box$}\vspace{.5pc}}

\begin{document}

\begin{frontmatter}

\title{Robustness: a New Form of Heredity Motivated by Dynamic Networks}

\author{Arnaud Casteigts}
\ead{arnaud.casteigts@labri.fr}
\address{Université de Bordeaux, CNRS,  Bordeaux INP, LaBRI, UMR 5800, France}

\author{Swan Dubois}
\ead{swan.dubois@lip6.fr}
\address{Sorbonne Universit\'e, CNRS, Inria, LIP6 UMR 7606, France}

\author{Franck Petit}
\ead{franck.petit@lip6.fr}
\address{Sorbonne Universit\'e, CNRS, Inria, LIP6 UMR 7606, France}

\author{John M. Robson}
\ead{robson@labri.fr}
\address{Université de Bordeaux, CNRS,  Bordeaux INP, LaBRI, UMR 5800, France}

\date{}

\setlength\textfloatsep{10pt plus 2pt minus 2pt}
\pagestyle{plain}



\begin{abstract}
We investigate a special case of hereditary property in graphs, referred to as {\em robustness}. A property (or structure) is called
robust in a graph $G$ if it is inherited by all the connected spanning subgraphs of $G$. We motivate this
definition using two different settings of dynamic networks. The first corresponds to networks of low dynamicity, where some links may be permanently removed so long as the network remains connected. The second corresponds to highly-dynamic networks, where communication links appear and disappear arbitrarily often, subject only to the requirement that the entities are temporally connected in a recurrent fashion ({\it i.e.} they can always reach each other through temporal paths). Each context induces a different interpretation of the notion of robustness.

We start by motivating the definition and discussing the two interpretations, after what we consider the notion independently from its interpretation, taking as our focus the robustness of {\em maximal independent sets} (MIS). A graph may or may not admit a robust MIS. We characterize the set of graphs \forallMIS in which {\em all} MISs are robust. Then, we turn our attention to the graphs that {\em admit} a robust MIS (\existsMIS). This class has a more complex structure; we give a partial characterization in terms of elementary graph properties, then a complete characterization by means of a (polynomial time) decision algorithm that accepts if and only if a robust MIS exists. This algorithm can be adapted to construct such a solution if one exists. 
\end{abstract}


\end{frontmatter}

\section{Introduction}
\label{sec:introduction}

The area of dynamic networks covers a variety of contexts, ranging from nearly-static networks where the network topology changes only occasionally, to highly-dynamic settings where the entities interact in a volatile way, through communication links which appear and disappear arbitrarily often and unexpectedly. In the second case, the immediate structure of the communication graph at given time (\ie its structural snapshot) does not capture much information, the main features being rather of a temporal nature. For example, the snapshots may never be connected, and yet offer a form of connectivity over time and space, called {\em temporal connectivity}. A number of formalisms were proposed recently to capture the temporality of these contexts, such as evolving graphs, time-varying graphs, link streams, and temporal graphs (see \eg~\cite{Fer04,CFQS12,LVM17,KKK02}, among many others). In this article, we introduce a graph concept which is strongly motivated by the highly-dynamic setting; however, the notion itself can be formulated and studied in terms of standard graphs, independently from its temporal interpretation.

Given a (standard) graph $G$, a given property (or structure) in $G$ is said to be {\em robust} if and only if it is inherited by every {\em connected spanning subgraph} of $G$. Hereditary properties based on the removal of edges are generally called {\em monotone} (see for instance~\cite{heredity1,heredity2,K88}). Robustness is therefore a particular case of monotonicity, in which the subgraph is additionally constrained to remain spanning and connected. (This concept is different from other uses of the term ``robustness'' in the literature, see \eg~\cite{robustness3,robustness1,robustness2}.) 

As explained, robustness can be interpreted in several different ways. (1) {\it Static networks with permanent link crashes:} Here, the network is essentially static; however,
it deteriorates over time and some of the links may definitively stop working (or equivalently, be removed). 
The network is to be used so long as it still connects all the nodes. It is easy to see, that if a property is robust in such a network, then it will remain valid so long as the network is used, despite the uncertainty regarding which of the links will crash. (2)~{\it Recurrent temporal connectivity in highly-dynamic networks:} Here, the immediate structure of the network does not matter as much as its temporal properties. In particular, a basic assumption is that temporal paths exists recurrently between all the entities, corresponding to Class~{$\cal TC^R$} in~\cite{Cas18} and Class~5 in~\cite{CFQS12}. Dubois et al. observed in~\cite{DKP15} that this assumption is equivalent to the guarantee that a subset of the edges always reappears, although in general such as subset is not known in advance. Thus, robustness can be interpreted in $\cal TC^R$ as the fact that a property is satified with respect to the subset of recurrent edges, whichever these edges are. The choice for a given interpretation is not mandatory, as robustness can be studied independently. However, the second interpretation being the original motivation here, we develop it further in a dedicated (but optional) section in the end of the paper.

In this paper, we illustrate the concept of robustness through a classical covering problem called {\em maximal independent set}
(MIS), which consists of selecting a subset of vertices none of which are
neighbors (independence) and which is maximal for inclusion.
Let us first observe that robust MISs may or may not exist depending on the considered graph. 
For example, if the graph is a triangle, then only one such structure exists up to isomorphism, consisting of a single
vertex (refer to Figure~\ref{fig:mis-a}). If an edge next to the selected vertex is removed, then this set is no longer
maximal. Therefore, the triangle graph admits no robust MIS.
\begin{figure}
  \centering
  \begin{minipage}[b]{.16\linewidth}
    \begin{tikzpicture}
      \tikzstyle{every node}=[circle, inner sep=1.2pt, fill=darkgray]
      \path (0,0) node (a) {};
      \path (a)+(0:1) node (b) {};
      \path (a)+(-60:1) node (c) {};
      
      \draw (a)--(b)--(c)--(a);
      
      \tikzstyle{every node}=[circle, draw, red, inner sep=2pt]
      \path (a) node {};
    \end{tikzpicture}
    \subcaption{}\label{fig:mis-a}
  \end{minipage}
  \begin{minipage}[b]{.26\linewidth}
    \begin{tikzpicture}
      \tikzstyle{every node}=[circle, inner sep=1.2pt, fill=darkgray]
      \path (0,0) node (a) {};
      \path (a)+(-30:1) node (b) {};
      \path (b)+(0:1) node (c) {};
      \path (b)+(-60:1) node (d) {};
      \path (c)+(30:1) node (e) {};
      
      \draw (a)--(b)--(c)--(e);
      \draw (b)--(d)--(c);
      
      \tikzstyle{every node}=[circle, draw, red, inner sep=2pt]
      \path (a) node {};
      \path (c) node {};
    \end{tikzpicture}
    \subcaption{}\label{fig:mis-b}
  \end{minipage}
  \begin{minipage}[b]{.26\linewidth}
    \begin{tikzpicture}
      \tikzstyle{every node}=[circle, inner sep=1.2pt, fill=darkgray]
      \path (0,0) node (a) {};
      \path (a)+(-30:1) node (b) {};
      \path (b)+(0:1) node (c) {};
      \path (b)+(-60:1) node (d) {};
      \path (c)+(30:1) node (e) {};
      
      \draw (a)--(b)--(c)--(e);
      \draw (b)--(d)--(c);
      
      \tikzstyle{every node}=[circle, draw, red, inner sep=2pt]
      \path (a) node {};
      \path (d) node {};
      \path (e) node {};
    \end{tikzpicture}
    \subcaption{}\label{fig:mis-c}
  \end{minipage}
  \begin{minipage}[b]{.18\linewidth}
    \begin{tikzpicture}
      \tikzstyle{every node}=[circle, inner sep=1.2pt, fill=darkgray]
      \path (0,0) node (a) {};
      \path (a)+(0:1) node (b) {};
      \path (b)+(-90:1) node (c) {};
      \path (c)+(-180:1) node (d) {};
      
      \draw (a)--(b)--(c)--(d)--(a);
      
      \tikzstyle{every node}=[circle, draw, red, inner sep=2pt]
      \path (a) node {};
      \path (c) node {};
    \end{tikzpicture}
    \subcaption{}\label{fig:mis-d}
  \end{minipage}
  \caption{\label{fig:examples} Four examples of MISs in various graphs.}
\end{figure}
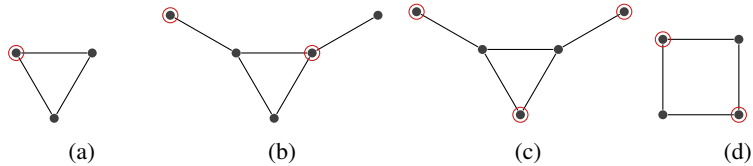
Some graphs admit both robust and non-robust MISs, as exemplified by the bull graphs on Figures~\ref{fig:mis-b}
(non-robust) and~\ref{fig:mis-c} (robust).
Finally, some graphs like the square graph (Figure~\ref{fig:mis-d}) are such that {\em all} MISs are robust. 

\subsection{Contributions}
In addition to the concept itself and the related discussions, we characterize exactly the set \forallMIS of the graphs in which all MISs are robust. 
To this end, we first define a class of graphs called {\em sputniks} (for reasons that will become clear later). 
Sputniks include, among others, all the trees, for which every property is trivially robust (since none
of the edges are removable). We show that \forallMIS consists exactly of the union of sputniks and complete
bipartite graphs. The interest of this {\em universal} class is that
finding a robust MIS in it amounts to finding a standard MIS.

The {\em existential} versions of this class, namely the set \existsMIS of those graphs which admit a robust solution seem to have a more complex structure.
We first give a sufficient condition and we show that this condition is necessary in the particular case of biconnected graphs (meaning here 2-vertex-connected). 
The more general case is addressed by means of an algorithm that decides if a given graph belongs to \existsMIS.
The trivial strategy for such an algorithm would amount to enumerating all MISs until a robust one is found. However, exponentially many MISs may
exist in general graphs (see Moon and Moser~\cite{MM65}, and~\cite{F87,GGG88} in the particular case of
connected graphs) and the validity of each one may have to be satisfied in exponentially many connected spanning
subgraphs. Motivated by this observation, we present a polynomial time decision algorithm, which can be adapted into a constructive algorithm (without significant overhead). Our algorithm relies on a particular decomposition of the graph into a tree
of 2-vertex-connected (biconnected) components called an ${\cal ABC}$-tree (a variant of block-cut
trees~\cite{Harary}), along which constraints are propagated as to the MIS status of intermediate vertices. The inner
constraints of non-trivial components are solved by a reduction to the 2-SAT problem. The yes-instances of this algorithm characterize \existsMIS, albeit indirectly. 
Whether \existsMIS admits a more elementary characterization in terms of graph properties is left open.

\subsection{Organization of the paper} Section~\ref{sec:definitions} presents the main definitions and concepts. 
Next, Section~\ref{sec:forallMIS} presents the characterization of class \forallMIS. 
In Section~\ref{sec:biconnectivity}, we show that if the graph is biconnected, then being bipartite is both necessary and sufficient to belong to \existsMIS. Then, we present the decision algorithm in Sections~\ref{sec:algo} through~\ref{sec:actual-algorithm}, its complexity in Section~\ref{sec:complexity} and its adaptation into a constructive algorithm in Section~\ref{sec:constructivity}.
In Section~\ref{sec:temporal-interpretation}, we develop the discussion regarding the temporal interpretation of robustness, motivated by highly-dynamic networks. We conclude in Section~\ref{sec:conclusion} by observing some additional features of robustness in general and a few open questions.

\section{Main Concepts and Basic Results}
\label{sec:definitions}

Let $G=(V,E)$ be a simple connected undirected graph on a finite set $V$ of $n$ vertices (or nodes). 
We denote by $N(v)$ the neighbors of vertex $v$, \ie the set $\{w \mid \{v,w\}\in E\}$. The degree of a vertex $v$ is $|N(v)|$. A vertex is {\em pendant} if it has degree~$1$. An {\em articulation point} (or {\em cut vertex}) is a vertex whose removal disconnects the graph. A {\em bridge} (or {\em cut edge}) is an edge whose removal disconnects the graph. 
We say that an edge is {\em removable} in a graph $G$ if it is not a bridge of $G$. Given $V' \subseteq V$, the induced subgraph $G[V']$ is the graph whose vertex set is $V'$ and whose edge set consists of all of the edges in $E$ that have both endpoints in $V'$.
In the context of this paper, a graph is said to be {\em biconnected} if it has at least three vertices, and it remains connected after the removal of any single vertex (\ie 2-vertex-connectivity). A biconnected {\em component} is a maximal biconnected subgraph.
Finally, a {\em spanning connected subgraph} of a graph $G=(V,E_G)$ is a graph $H=(V, E_H)$ such that $E_H \subseteq E_G$, and $H$ is connected. 
We define the concept of robustness as follows.

\begin{definition}[Robustness]
\label{def:robustness}
A property $P$ is {\em robust} in $G$ if and only if it is satisfied in every connected spanning subgraph of $G$ (including $G$ itself).
\end{definition}

In other words, a robust property holds even after an arbitrary number of edges are removed without disconnecting the graph. Robustness is a special case of hereditary property, and more precisely a special case of a decreasing monotone property (see for instance~\cite{K88}). 
The term ``property'' includes both basic graph properties and solutions to combinatorial problems. Our focus in this initial work is on the latter; however, looking at the robustness of basic graph properties might help understand this notion further, for instance, connectivity itself is a trivial robust property. Bipartiteness is also a robust that is discussed at several occasions in this paper. 

In the present work, we focus on the {\em maximal independent sets} (MIS) problem, which consists of selecting a subset of vertices none of which are neighbors (independence) and to which no further vertex can be added (maximality). 
Following Definition~\ref{def:robustness}, a robust MIS (RMIS, for short) in a graph $G$ is a subset of vertices which remains {\em maximal} and {\em independent} in every connected spanning subgraph of $G$. 

\begin{observation}
  \label{obs:independence}
The notion of independence is stable under the removal of edges.  Therefore, it is sufficient that an MIS remains maximal in order to be an RMIS.
\end{observation} 

Let us define the following two classes of graphs.

\begin{definition}[\forallMIS]
  Set of graphs in which all MISs are robust.
\end{definition}

\begin{definition}[\existsMIS]
  Set of graphs that admit at least one robust MIS.
\end{definition}

Trivially, $\forallMIS \subseteq \existsMIS$.
Finally, we define two classes of graphs which play a central role in this work, namely {\em complete bipartite graphs} and {\em sputnik graphs}, the latter being new.
\begin{definition}[Complete bipartite graph]
A complete bipartite graph is a graph $G=(V_1 \cup V_2,E)$ such that $V_1 \cap V_2 = \emptyset$ and $E = \{\{v_1,v_2\},v_1\in V_1\text{ and }v_2\in V_2\}$. 
\end{definition}

In other words, the vertices can be partitioned into two sets $V_1$ and $V_2$ such that every vertex in $V_1$ shares an edge with every vertex in $V_2$ (completeness), and these are the only edges (bipartiteness).

\begin{definition}[Sputnik]
\label{def:sputnik}
A graph is a {\em sputnik} if and only if every vertex $v$ belonging to a cycle has at least one pendant neighbor.
\end{definition}

An example of sputnik is shown on Figure~\ref{fig:sputnik}. This name was chosen by analogy with the well-known satellite. By the same analogy, we say that a vertex has an {\em antenna} if it has at least one pendant neighbor. 

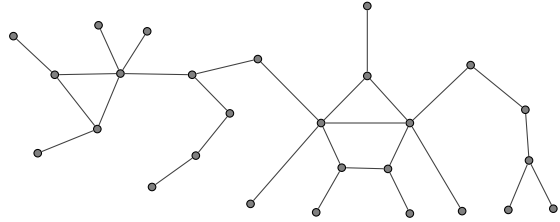
\begin{figure}[h]
\begin{center}
\begin{tikzpicture}[xscale=.8,yscale=.8]
  \tikzstyle{every node}=[defnode, inner sep=1pt]
  \path (1.5,10.3) node [] (v0) {};
  \path (2.9,10.48) node [] (v3) {};
  \path (3.7,10.38) node [] (v26) {};
  \path (1.9,8.36) node [] (v5) {};
  \path (3.78,7.8) node [] (v9) {};
  \path (6.48,7.38) node [] (v16) {};
  \path (8.02,7.36) node [] (v17) {};
  \path (7.32,10.8) node [] (v18) {};
  \path (5.4,7.52) node [] (v19) {};
  \path (8.88,7.4) node [] (v20) {};
  \path (9.64,7.42) node [] (v24) {};
  \path (10.38,7.44) node [] (v25) {};

  \path (2.18,9.66) node [] (v1) {}; 
  \path (3.26,9.68) node [] (v2) {}; 
  \path (2.88,8.76) node [] (v4) {}; 
  \path (6.56,8.86) node [] (v11) {}; 
  \path (7.32,9.64) node [] (v12) {}; 
  \path (8.02,8.86) node [] (v13) {}; 
  \path (6.9,8.12) node [] (v14) {}; 
  \path (7.66,8.1) node [] (v15) {}; 
  \path (4.5,8.32) node [] (v8) {}; 
  \path (9.98,8.24) node [] (v23) {}; 

  \path (4.44,9.66) node [] (v6) {}; 
  \path (9.02,9.82) node [] (v21) {}; 
  \path (5.06,9.02) node [] (v7) {}; 
  \path (5.52,9.92) node [] (v10) {}; 
  \path (9.92,9.08) node [] (v22) {}; 
  \tikzstyle{every path}=[];
  \draw [darkgray] (v0)--(v1);
  \draw [darkgray] (v1)--(v2);
  \draw [darkgray] (v2)--(v4);
  \draw [darkgray] (v4)--(v1);
  \draw [darkgray] (v5)--(v4);
  \draw [darkgray] (v2)--(v3);
  \draw [darkgray] (v2)--(v26);
  \draw [darkgray] (v2)--(v6);
  \draw [darkgray] (v7)--(v6);
  \draw [darkgray] (v9)--(v8);
  \draw [darkgray] (v8)--(v7);
  \draw [darkgray] (v6)--(v10);
  \draw [darkgray] (v10)--(v11);
  \draw [darkgray] (v11)--(v12);
  \draw [darkgray] (v12)--(v13);
  \draw [darkgray] (v13)--(v11);
  \draw [darkgray] (v11)--(v14);
  \draw [darkgray] (v14)--(v15);
  \draw [darkgray] (v15)--(v13);
  \draw [darkgray] (v16)--(v14);
  \draw [darkgray] (v17)--(v15);
  \draw [darkgray] (v18)--(v12);
  \draw [darkgray] (v19)--(v11);
  \draw [darkgray] (v20)--(v13);
  \draw [darkgray] (v21)--(v13);
  \draw [darkgray] (v22)--(v21);
  \draw [darkgray] (v23)--(v22);
  \draw [darkgray] (v23)--(v24);
  \draw [darkgray] (v23)--(v25);
\end{tikzpicture}
\end{center}
\caption{\label{fig:sputnik}Example of a sputnik graph (whose planarity is accidental).}
\end{figure}

\section{Characterization of \forallMIS}
\label{sec:forallMIS}

In this section, we show that the class of graphs in which all MISs are robust, \forallMIS, corresponds exactly to the
union of complete bipartite graphs and sputnik graphs. 
%
%
We first show that all the MISs in a complete bipartite graph are robust, and so are all MISs in a sputnik graph (sufficient conditions). 
Next, we show that these graphs are the only ones (necessary condition). 
The necessary side is more complex because two classes of graphs are involved. 
The proof proceeds by showing that if all MISs are robust in a graph which is not a sputnik, then that graph {\em must} be a complete bipartite graph.

\begin{lemma}
  \label{lem:BK}
  All MISs in a complete bipartite graph are robust.
\end{lemma}

\begin{proof}
From Observation~\ref{obs:independence}, we only need to show that maximality is preserved.
There are two ways of chosing an MIS in a complete bipartite graph $G=(V_1\cup V_2, E)$, namely $V_1$ or $V_2$. Without loss of generality, let $V_1$ be chosen. Then, in all connected spanning subgraphs of $G$, every vertex in $V_2$ still has at least one neighbor in $V_1$ (otherwise the graph is disconnected), which preserves maximality.
\end{proof}

\begin{lemma}
  \label{lem:sputniks}
  All MISs in a sputnik graph are robust.
\end{lemma}

\begin{proof}
Again, by Observation~\ref{obs:independence}, we only need to show that maximality is preserved. 
By definition of a sputnik graph (Definition~\ref{def:sputnik}), if an edge $\{u,v\}$ is removable (\ie it belongs to a cycle), then both of its endpoints have an antenna. Maximality implies that either $u$ or the tip of its antenna are in the set. (The
same holds for $v$.) As antennas are bridges, they cannot be removed and thus maximality is preserved when edges from a cycle are removed.
\end{proof}


\begin{lemma} 
  \label{lem:necessary}
  If $G$ is not a sputnik, and every MIS in $G$ is robust, then $G$ must be a complete bipartite graph.
\end{lemma}

\begin{proof}
	If $G$ is not a sputnik, then some vertex $u$ belonging to a cycle $C \subseteq G$ has no antenna. Consider the graph $G \setminus \{u\}$ and call $X_1, X_2, \dots, X_k$ the components that would result, except that a copy of $u$ is added back to each (see Figure~\ref{fig:claim} for an illustration). Without loss of generality, let $X_1$ be the one containing $C$. 
The other components (if any) are such that every neighbor of $u$ has another neighbor than $u$ (otherwise $u$ would have a pendant neighbor). 

\noindent\textbf{Claim.} If all MISs in $G$ are robust, then all neighbors of $u$ in $X_1$ have the same set of neighbors.

We prove this claim by contradiction.
Let two neighbors $v_1, v_2$ of $u$ be such that $N(v_1) \ne N(v_2)$. We will show that at least one MIS is not robust
(see Figure~\ref{fig:claim} for an illustration). Without loss of generality, let $x$ be a vertex in $N(v_1) \setminus N(v_2)$. Then an MIS can be built which contains both $v_2$ and $x$ (as a special case, $x$ may be the same vertex as $v_2$, which is not a problem). For each of the components $X_{i\ge 2}$, choose an edge $\{u,w_i\} \in X_i$ and add another neighbor of $w_i$ to the MIS (such a neighbor exists, as we have already seen). One can see that $u$, $v_1$ and all $w_i$ can no longer enter the MIS because they all have neighbors in it. Now, choose the remaining elements of the MIS arbitrarily. 
Now, consider the following removals: in all components $X_{i\ge 2}$ all edges incident to $u$ except $\{u,w_i\}$ are
	removed; and in $X_1$, all edges incident to $u$ except $\{u,v_1\}$ are removed. The resulting graph is
	connected since all of the $X_i \setminus \{u\}$ are connected, but the set is no longer maximal because $u$
	could now be added in it.\hfill {\it (End of Claim)}\\

\newcommand{\convexpath}[2]{
  [   
  create hullcoords/.code={
    \global\edef\namelist{#1}
    \foreach [count=\counter] \nodename in \namelist {
      \global\edef\numberofnodes{\counter}
      \coordinate (hullcoord\counter) at (\nodename);
    }
    \coordinate (hullcoord0) at (hullcoord\numberofnodes);
    \pgfmathtruncatemacro\lastnumber{\numberofnodes+1}
    \coordinate (hullcoord\lastnumber) at (hullcoord1);
  },
  create hullcoords
  ]
  ($(hullcoord1)!#2!-90:(hullcoord0)$)
  \foreach [
  evaluate=\currentnode as \previousnode using \currentnode-1,
  evaluate=\currentnode as \nextnode using \currentnode+1
  ] \currentnode in {1,...,\numberofnodes} {
    let \p1 = ($(hullcoord\currentnode) - (hullcoord\previousnode)$),
    \n1 = {atan2(\y1,\x1) + 90},
    \p2 = ($(hullcoord\nextnode) - (hullcoord\currentnode)$),
    \n2 = {atan2(\y2,\x2) + 90},
    \n{delta} = {Mod(\n2-\n1,360) - 360}
    in 
    {arc [start angle=\n1, delta angle=\n{delta}, radius=#2]}
    -- ($(hullcoord\nextnode)!#2!-90:(hullcoord\currentnode)$) 
  }
}

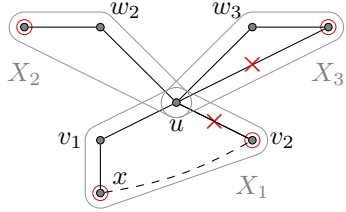
\begin{figure}[h]
  \centering
  \begin{tikzpicture}
    \tikzstyle{every node}=[defnode, inner sep=1pt]
    \path (0,1) node (a){};
    \path (1,1) node (w2){};
    \path (2,0) node (u){};
    \path (3,1) node (w3){};
    \path (4,1) node (b){};
    \path (1,-.5) node (v1){};
    \path (3,-.5) node (v2){};
    \path (1,-1.2) node (x){};
    \tikzstyle{every node}=[circle, draw, red, inner sep=2pt]
    \path (a) node{};
    \path (b) node{};
    \path (v2) node{};
    \path (x) node{};
    \tikzstyle{every node}=[]
    \path (u) node[below=3pt] {$u$};
    \path (w2) node[above right] {$w_2$};
    \path (w3) node[above left] {$w_3$};
    \path (v1) node[left=3pt] {$v_1$};
    \path (v2) node[right=3pt] {$v_2$};
    \path (v2) node[gray,below=10pt] {$X_1$};
    \path (a) node[gray,below=10pt] {$X_2$};
    \path (b) node[gray,below=10pt] {$X_3$};
    \path (x) node[above right=1pt,yshift=-1pt] {$x$};
    \draw (a)--(w2)--(u)--(w3)--(b);
    \draw (u) edge node[midway,red] {$\boldsymbol \times$} (b);
    \draw (u) edge node[midway,red] {$\boldsymbol \times$} (v2);
    \draw (x)--(v1)--(u)--(v2);
    \draw (x) edge[dashed, bend right=10] (v2);

    \tikzstyle{every path}=[inner sep=10pt]
    \draw[gray!80] \convexpath{a,w2,u}{.2cm};
    \draw[gray!80] \convexpath{u,w3,b}{.2cm};
    \draw[gray!80] \convexpath{u,v2,x,v1}{.2cm};
  \end{tikzpicture}
  \caption{\label{fig:claim}Contradictory example for proving the claim, where $x \in N(v_1)\setminus N(v_2)$. The dashed line represents the rest of component $X_1$.}
\end{figure}
        
Since $u$'s neighbors in $X_1$ have the same neighbors, this means in particular that none of these vertices has an
antenna. As a result, the arguments that applied to $u$ because of its absence of antenna, apply in turn to
$u$'s neighbors in $X_1$. In particular, the neighbors of these vertices (including $u$) must have the same
set of neighbors. This implies that (1) $u$ can no longer be an articulation point, thus $G=X_1$ and (2) all
neighbors of $u$ have the same set of neighbors, and these neighbors also have the same set of neighbors,
which finally implies that the graph is complete bipartite.
\end{proof}

Lemmas~\ref{lem:BK}, \ref{lem:sputniks}, and \ref{lem:necessary} allow us to conclude with
Theorem~\ref{th:forallMIS}~:

\begin{theorem}
  \label{th:forallMIS}
  All MISs are robust in a graph $G$ if and only if $G$ is complete bipartite or sputnik.
\end{theorem}

\section{Characterization of \existsMIS}
\label{sec:existsMIS}
In this section, we turn our attention to the characterization of graphs which {\em admit} a robust MIS ($\existsMIS$).
Unfortunately, this class does not seem to admit a simple characterization in terms of elementary graph properties. We start by discussing some such properties that give a partial characterization, namely we show that bipartiteness is a necessary and sufficient condition for {\em biconnected} graphs to be in \existsMIS. Then, we turn our attention to the general case and present an algorithm that decides if a given graph admits a robust MIS (in polynomial time, and constructively). The yes-instances of this algorithm are an indirect characterization of \existsMIS. 

\subsection{Bipartiteness versus Biconnectivity}
\label{sec:biconnectivity}

Let us first observe that bipartiteness is a sufficient condition for any graph to admit a robust MIS.

\begin{lemma}
  \label{obs:bipartiteness}
  All bipartite graphs admit a robust MIS. 
\end{lemma}

\begin{proof}
  The argument generalizes that of Lemma~\ref{lem:BK} for complete bipartite graphs. Let an MIS be composed of all the
	vertices of the first part. As long as the graph remains connected, every vertex in the second part has a
	neighbor in the first part, and thus in the MIS, which preserves maximality. (Independence is not
	impacted---Observation~\ref{obs:independence}.)
\end{proof}

In fact, bipartiteness happens to be also a necessary condition in the
particular case of biconnected graphs.

\begin{lemma}\label{lem:nonbi}
  If a biconnected graph $G$ is not bipartite, then it admits no robust MIS. 
\end{lemma}
\begin{proof}
By contradiction, suppose that $G$ admits a robust MIS. As $G$ is not bipartite, it is not 2-colorable, thus either two neighbor vertices exist which are both included in the set, or two neighbor vertices exist which are both excluded from the set. In the first case, independence is contradicted. In the second case, let $u$ and $v$ be two such vertices. Because the graph is
	biconnected, it is possible to remove all the incident edges to $u$ except $\{u,v\}$ without disconnecting $G$, resulting in a non-maximal set, thus contradicting robustness.
\end{proof}

The argument in the proof of Lemma~\ref{lem:nonbi} will be used several times in the rest of the paper. We refer to it through the concept of a weak vertex.

\begin{definition}[Weak vertex]
  \label{def:weak}
  Let $u$ be a vertex that is not included in an MIS. If $u$ has another neighbor $v$ not being included in the MIS and such that all edges incident to $u$ except $\{u,v\}$ can be simultaneously removed without disconnecting the graph, then $u$ is called a {\em weak} vertex.
\end{definition}

Lemmas~\ref{obs:bipartiteness} and~\ref{lem:nonbi} imply that the intersection of $\existsMIS$ with biconnected graphs is {\em exactly} the set of bipartite biconnected graphs (see Table~\ref{tab:dichotomy}).
\begin{table}[h]
  \centering
  \begin{tabular}{c|c|c}
    &Bipartite&$\neg$ Bipartite\\\hline
    Biconnected&yes& no\\\hline
    $\neg$ Biconnected&yes&possibly\\
  \end{tabular}
  \caption{\label{tab:dichotomy} Existence of a robust MIS with respect to bipartiteness and biconnectedness.}
\end{table}
 An immediate
corollary is the existence of an infinite family of graphs which do not admit a robust MIS, namely all
biconnected non-bipartite graphs (including, for example, the triangle
graph in Figure~\ref{fig:examples} and more generally all the odd cycles).

\subsection{Overview of the algorithm}
\label{sec:algo}


The problem of computing a {\em standard} MIS in a graph $G=(V_G,E_G)$ can be solved by a one-sentence greedy algorithm as follows: for all vertices $v \in V_G$ in arbitrary order, include $v$ in the MIS if none of its neighbors already is.
The problem of computing a {\em robust} MIS (or RMIS) is fundamentally different in two respects: (1) Solutions may or may not exist, and (2) Even if a solution exists, a decision made in some part of the graph can restrict (or invalidate) the feasible choices in remote parts. For example, in the graph of Figure~\ref{fig:non-local}, if vertex $v$ is included in the set, then an RMIS exists if and only if node $u$ is not included in the set. (Any other choice would produce a weak vertex.)

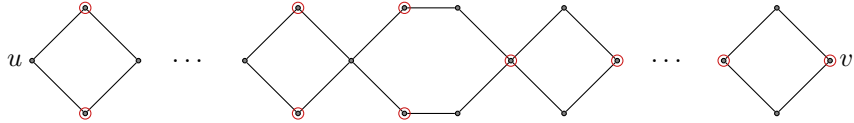
\begin{figure}[h]
  \centering
\begin{tikzpicture}[scale=.7]
  \tikzstyle{every node}=[defnode]
  \path (-1,1) node [] (ma) {};
  \path (0,0) node [] (ba) {};
  \path (0,2) node [] (ha) {};
  \path (1,1) node [] (mb) {};
  \path (3,1) node [] (m0) {};
  \path (4,0) node [] (b0) {};
  \path (4,2) node [] (h0) {};
  \path (5,1) node [] (m1) {};
  \path (6,0) node [] (b3) {};
  \path (6,2) node [] (h3) {};
  \path (7,0) node [] (b3b) {};
  \path (7,2) node [] (h3b) {};
  \path (8,1) node [] (m4) {};
  \path (9,0) node [] (b4) {};
  \path (9,2) node [] (h4) {};
  \path (10,1) node [] (m5) {};
  \path (12,1) node [] (m7) {};
  \path (13,0) node [] (b7) {};
  \path (13,2) node [] (h7) {};
  \path (14,1) node [] (m8) {};

  \tikzstyle{every node}=[]
  \path (ma) node[left] {$u$};
  \path (m8) node[right] {$v$};
  \path (mb) node[right=9pt] {$\dots$};
  \path (m5) node[right=9pt] {$\dots$};

  \draw (ma)--(ba)--(mb)--(ha)--(ma);
  \draw (m0)--(b0)--(m1)--(h0)--(m0);
  \draw (m1)--(b3)--(b3b)--(m4)--(h3b)--(h3)--(m1);
  \draw (m4)--(b4)--(m5)--(h4)--(m4);
  \draw (m7)--(b7)--(m8)--(h7)--(m7);

  \tikzstyle{every node}=[mis]
  \path (ha) node {};
  \path (ba) node {};
  \path (h0) node {};
  \path (b0) node {};
  \path (h3) node {};
  \path (b3) node {};
  \path (m4) node {};
  \path (m5) node {};
  \path (m7) node {};
  \path (m8) node {};
\end{tikzpicture}
\caption{\label{fig:non-local}Non-locality of finding a robust MIS.}
\end{figure}

More generally, deciding whether a graph admits an RMIS requires the identification and propagation of constraints within the graph. To do so, our algorithm relies on a particular type of decomposition of the input graph as a tree of biconnected components called \ABCT. The constraints are first determined at the leaves of this tree, then they are propagated and modified upward, until the root component is itself analysed. At an intermediate node of this tree, either the corresponding subtree admits an RMIS or it does not. If it does, a condition may apply regarding the status of the topmost vertex in the subtree, e.g. the subtree may admit an RMIS at the condition that this vertex is (or is not) in the MIS.

In the following, we always refer to the vertices of the input graph as {\em vertices}, and to those of the decomposition tree as {\em nodes} to avoid confusion.
Note that if $G$ is a tree, then none of its edges can be removed (thus all MISs are robust). As a result, we restrict our attention to the cases that $G$ has at least one biconnected component. We also assume that $G$ is connected, as otherwise each part can be solved separately.

\subsection{The \ABCT decomposition}

An \ABCT, denoted by $T=(V_T,E_T)$, is neither a block-cut tree, nor a bridge tree (see~\cite{Harary} for background), it is a mixture of both. Precisely, an \ABCT is made of four types of nodes $\A, \B, \C$ and $\PV$ defined with respect to the input graph $G=(V_G,E_G)$ as follows:
\begin{itemize}
\item \PV is the set of pendant vertices,
\item \A is the set of articulation points,
\item \B is the set of bridges, 
\item \C is the set of biconnected component.
\end{itemize}
Thus, every node in $\PV$ and $\A$ corresponds to an original {\em vertex} in $V_G$ (the converse is not true), every node in $\B$ corresponds to an {\em edge} in $E_G$ (same remark), and every node in $\C$ corresponds to an entire {\em subgraph} of~$G$ (same remark again). Considering the graph in Figure~\ref{fig:treecomp}, one would obtain $\PV=\{4,5,7,12,20,24\}$, $\A=\{2,3,6,8,10,11,14,15,16,17,18,21,22,28\}$, $\B = \{\{2,3\},\{3,4\},\{4,5\},\{6,7\},\dots\}$, and
$\C = \{\texttt{A},\texttt{B},\texttt{C},\texttt{D},\texttt{E}\}$.

Observe that a same vertex of $V_G$ may correspond to a node $\A$, and at the same time be the endpoint of one or several bridges in $\B$, and at the same time belong to one or several biconnected components in $\C$. Also observe that the endpoints of a bridge are always articulation points or pendant vertices. All these relations are materialized by the set of edges $E_T$ of the \ABCT, defined as

\begin{center}
  $E_T= \{\{a,X\} \in \A\times\C \mid a \in X\}$ $\cup$ $\{\{v,Y\} \in (\A\cup\PV) \times \B \mid v \in Y\}$.
\end{center}

In words, articulation points (seen as nodes) share an edge with the biconnected components they belong to (if any), and articulation points or pendant nodes (seen as nodes) share an edge with the bridges they belong to (if any).
The \ABCT corresponding to the graph of Figure~\ref{fig:treecomp} is shown in Figure~\ref{fig:ex_compo}. The reader is encouraged to spend a few minutes getting acquainted with this construction, which is used frequently in the following.

\begin{figure}[h]
  \centering
  \includegraphics[width=.9\textwidth] {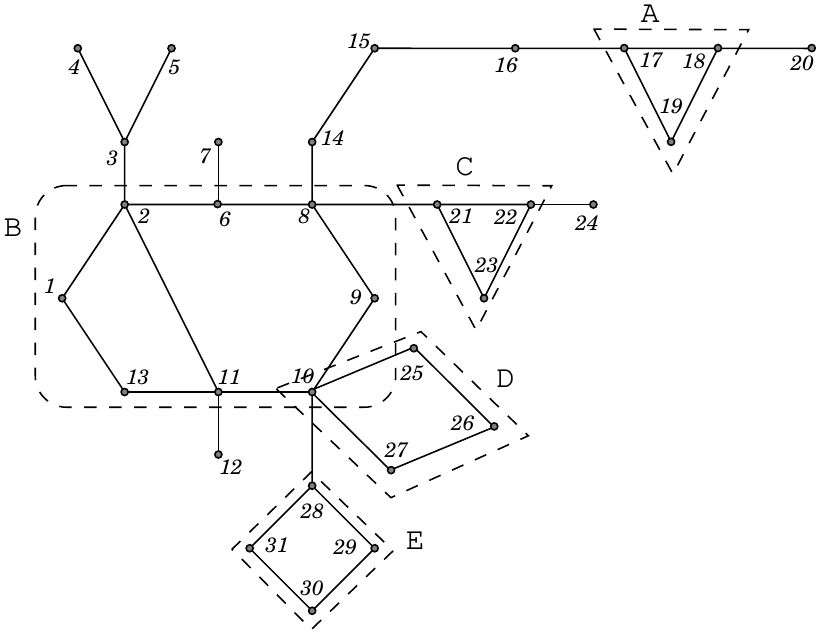}
  \caption{\label{fig:treecomp}
    Example of input graph, seen as a tree of biconnected components.
}
\end{figure}
\begin{figure}[h]
  \centering
  \includegraphics[width=1\textwidth]{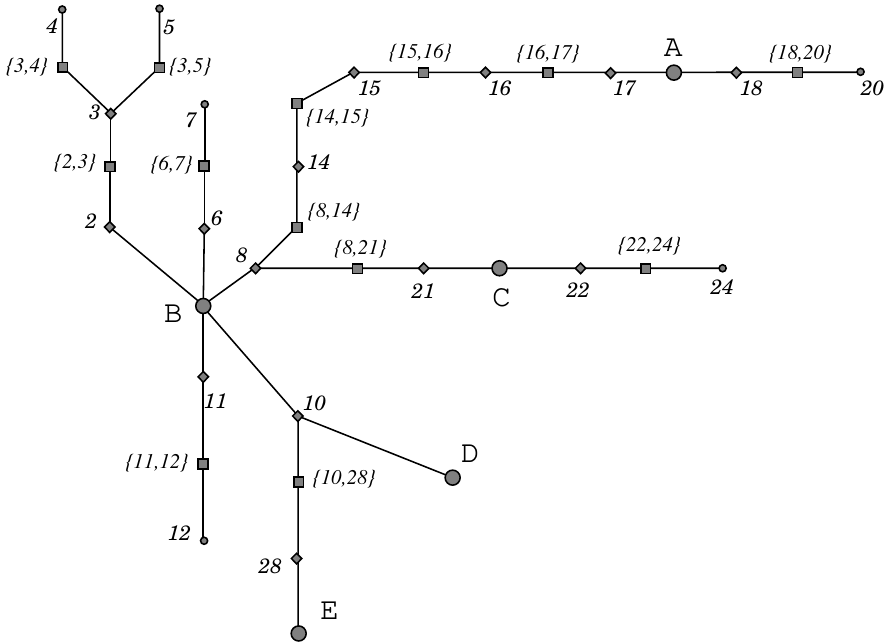}
  \vspace{-2.4cm}
  
  \hfill
  \begin{tikzpicture}
    \path (0,0) node[text width=5cm](legend){
      {\small $\circ$}: pendant vertices ($\in \PV$) \\
      {\scriptsize $\Diamond$}: articulation points ($\in \A$)\\
      {\tiny $\Box$}: bridge edges ($\in \B$)\\
      {\scriptsize $\bigcirc$}: biconnected components ($\in \C$)
    };
  \end{tikzpicture}
\caption{\label{fig:ex_compo}
  \ABCT corresponding to the input graph on Figure~\ref{fig:treecomp}.
}
\end{figure}

\subsection{RMIS constraints in a rooted \ABCT}
\label{sec:constraints}

The algorithm proceeds by propagating constraints from the leaves of the \ABCT to a root $\R$ chosen arbitrarily among the biconnected components. Given a node in $x \in V_T\setminus \R$, we denote by $\parent(x)$ the neighbor of $x \in V_T$ that is one-hop closer to \R, and by $\children(x)$ the set of nodes $\{y \in V_T\setminus\R \mid \parent(y)=x\}$. We extend these definitions to $\descendants()$ in the natural way. Because some nodes in $T$ correspond to real vertices (namely, the nodes in $\A$ and $\PV$) and others do not (the nodes in $\B$ and $\C$), we define a notion of {\em attachment vertex} $v(x)$ of a node $x$ as being either the underlying vertex itself (if $x\in\A\cup\PV$) or the underlying vertex of the parent (if $x\in\B\cup\C$). Indeed, the parent of a node in $\B\cup \C$ is always a node in $\A$, thus it corresponds to a real vertex in $V_G$. Intuitively, the attachment vertex of a node is the highest vertex in this node towards the root $R$.
For instance, in 
Figure~\ref{fig:ex_compo}, assuming that \R is component \texttt{D}, then we have for instance $v(4)=4$, $v(\{3,4\})=3$, $v(3)=3$, $v(\{8,14\})=v(\{8,21\})=v(8)=8$, and $v(\{10,28\})=v(\texttt{B})=v(10)=10$, and indeed, $10$ is the highest underlying vertex in all these nodes (with respect to the root).

Attachment vertices play a important role in the management of constraints, because the constraints induced by $x$ are ultimately aggregated as a membership status for $v(x)$ in the MIS.
The more formal definition relies on induced subgraphs as follows. Every node $x$ induces a subgraph $G_x \subseteq G$ which is exactly $G[x]$ when $x\in\A \cup\B\cup\PV$ (see definitions in Section~\ref{sec:definitions}), or $x$ itself when $x\in\C$. By extension, for any $x\in V_T$, the subtree $T_x=(V_{T_x},E_{T_x})\subseteq T$ whose highest node with respect to $\R$ is $x$ induces a subgraph $G[{T_x}]\subseteq G$ defined as $\cup G_y$ over all $y \in (\descendants(x)\cup x)$.
In other words, $G[{T_x}]$ is the subgraph of $G$ which corresponds to the subtree of $x$.
We define the concept of {\em aerial} version of $G[T_x]$, noted $G_a[T_x]$ as $G[T_x]\cup G[\{v(x),a\}]$ with $a$ an artificial neighbor of $v(x)$ that does not belong to $G[T_x]$. 
For any $x$, the constraints applying to $v(x)$ are encoded using the three following labels (also called tags):
\begin{itemize}
\item \typePI (\typePIdesc): $G[T_x]$ admits an RMIS that includes $v(x)$;
\item \typePO (\typePOdesc): $G[T_x]$ admits an RMIS that does not include $v(x)$;
\item \typePE (\typePEdesc): $G_a[T_x]$ admits an RMIS that does not include $v(x)$.
\end{itemize}

Observe that a node may have several labels at the same time. For example \typePI and \typePO, or \typePI and \typePE. On the other hand, we use \typePO and \typePE in a mutually exclusive way based on the following remark.

\begin{remark}
  \label{rem:POPE}
  If a node $x$ has label $\typePO$, then $G[T_x]$ admits an RMIS in which $v(x)$ is not included but one of its neighbor in $G[T_x]$ is (due to maximality). Thus, no inclusion constraint applies regarding the external neighbor of $v(x)$. As a result, whenever the constraints from different children induce both \typePO and \typePE, \typePO is chosen.
\end{remark}


\subsection{Constraint identification and propagation}
\label{sec:propagation}

In this subsection, we present the rules used for identifying and propagating constraints within the \ABCT. The purpose of the rules is to determine what labels $L(x)$ a node $x$ should take based on the labels of its children in the \ABCT. The validity of the rules is established gradually along their descriptions, based on the following definition of a correct labeling.


\begin{definition}[Correct labeling] \label{def:correct-labeling}
  A node $x$ (in the underlying rooted \ABCT $T$) is correctly labeled if
\begin{itemize}
\item $\typePI\in L(x)$ if and only if $G[T_x]$ admits an RMIS that includes $v(x)$;
\item $\typePO\in L(x)$ if and only if $G[T_x]$ admits an RMIS that does not include $v(x)$; and
\item $\typePE\in L(x)$ if and only if $\typePO\notin L(x)$ and $G_a[T_x]$ admits an RMIS that does not include $v(x)$.
\end{itemize}
\end{definition}

The rules are presented below based on the type of $x$ (namely, $\A,\B,\C,$ and $\PV$). They are illustrated in reference to the input graph $G=(V_G,E_G)$ in Figure~\ref{fig:treecomp} and the corresponding \ABCT $T=(V_T,E_T)$ in Figure~\ref{fig:ex_compo}, with root $\texttt{D}\in \C$.
For simplicity, when discussing about the construction of an MIS, if a vertex $v$ is not included in the MIS and none of its neighbor are included, we say that $v$ is {\em not covered}. The other vertices are {\em covered}. (This terminology is standard in the literature on covering problems.)


\subsubsection{Pendant nodes ($x \in \PV$)}
This case is the easiest, because nodes in $\PV$ have no children and their labels are always the same.


\begin{labrule}\label{rule:P}
  $L(x)$ is set to $\{\typePI,\typePE\}$.
\end{labrule}

In words, $v(x)$ may or may not be included to the MIS, but if it is not, then $v(y)$ should be, where $y$ is the bridge node such that $\parent(x)=y$.

\begin{lemma}\label{lem:labelnodeprmis}
  If $x\in \PV$, then Labeling rule~\ref{rule:P} produces a correct labeling of $x$.
\end{lemma}

\begin{proof}
  If $v(x)$ is included in the solution set, then this set (made of $v(x)$ alone) is clearly maximal and independent in $G[T_x]$ (which is also $v(x)$ alone), thus $x$ can be labeled \typePI. If it is not included, but $v(\parent(x))$ is included, then the resulting set is maximal and independent in $G_a[T_x]$ (\ie the graph made of a single edge between $x$ and $\parent(x)$), thus $x$ can be labeled \typePE. In both cases, the considered subgraph of $G$ is itself a tree, thus any valid MIS in it is robust, thus labels \typePI and \typePE are both valid. 
On the other hand, $x$ cannot be labeled \typePO because if $v(x)$ is not in the MIS, then the corresponding set is not maximal in $G[T_x]$ (which is reduced to the single vertex $v(x)$). 
\end{proof}

\subsubsection{Articulation points ($x \in \A$)}

By construction of the \ABCT, if $x\in A$, then all the children of $x$ are in $\B$ or in $\C$ and their attachment vertex is nothing but $v(x)$ itself. Thus, the constraints of each children already relates to $v(x)$, albeit individually. The rule consists of aggregating these constraints as follows.

\begin{labrule}\label{rule:A}If \typePI belongs to the labels of all the children of $x$, then \typePI is added to $L(x)$; if all the children of $x$ contain a label \typePE or \typePO, then two subcases arise: either none of them contains \typePO, in which case \typePE is added to $L(x)$, or at least one contains \typePO, in which case \typePO is added to $L(x)$.
\end{labrule}

The formal description of Labeling rule~\ref{rule:A} is given by Lines~6 to~13 in Algorithm~\vref{algo:labelnode} (the algorithm itself is discussed in a subsequent section). In our example, nodes $3$, $11$, and $18$ are all labeled \typePI and \typePO. 

\begin{lemma}\label{lem:labelnodearmis}
  If $x \in \A$ and all the nodes in $\children(x)$ are correctly labeled, then Labeling rule~\ref{rule:A} produces a correct labeling of $x$.
\end{lemma}

\begin{proof}
  Let us assume that all the nodes in $\children(x)$ are correctly labeled. The proof follows the same cases as the labeling rule. We first prove that the assigned labels are valid, then we prove that they cannot be assigned otherwise.
  \begin{itemize}
  \item If all the nodes in $\children(x)$ contain label \typePI, then for all $y\in\children(x)$, $G[T_y]$ admits a robust MIS that includes $v(y)$. But since $x\in \A$, we have $v(y)=v(x)$, thus $G[T_x]$ admits a robust MIS that includes $v(x)$. 
  \item If all the nodes in $\children(x)$ contain label \typePE or \typePO, then two possible subcases arise:
    \begin{enumerate}
      \item None of them contains \typePO. In this case, all of them contain label \typePE, meaning that for all $y\in \children(x), G[T_y]$ does not admit a robust MIS that excludes $v(y)$, but $G_a[T_y]$ does. Since $v(x)=v(y)$, we have that $G[T_x]$ does not admit a robust MIS that excludes $v(x)$ (a single child with label \typePE would actually be enough here), but $G_a[T_x]$ does (here, we need that {\em all} the children have label \typePE).
      \item At least one contains \typePO. In this case, at least one child $y$ is such that $G[T_y]$ admits an RMIS that excludes $v(y)$. Since $v(y)=v(x)$, at least one of the neighbors of $v(x)$ in $G[T_y]$ can be included in such an MIS, which satisfies the aerial constraints of the $G_a[T_{y'}]$ for all other children $y'$ labeled \typePE (if any). Thus $G[T_x]$ admits an RMIS that excludes $v(x)$ without requiring further aerial constraints above $x$.
    \end{enumerate}
  \end{itemize}

We focus now on the negative direction. 
If $G[T_x]$ does not admit a robust MIS that includes $v(x)$, then because $v(x)=v(y)$, there exists at least one $y\in \children(x)$ such that $G[T_y]$ does not admit a robust MIS that includes $v(y)$. As the labeling of $y$ is correct, $y$ is not labeled \typePI and hence $v(x)$ is not labeled \typePI by the rule.
Similarly, if $G[T_x]$ does not admit a robust MIS that excludes $v(x)$, then because $v(x)=v(y)$, the two cases above are possible. $(i)$ There exists at least one $y\in \children(x)$ such that neither $G[T_y]$ nor $G_a[T_y]$ admit a robust MIS that excludes $v(y)$. As the labeling of $y$ is correct, $y$ is not labeled \typePO or \typePE. $(ii)$ For any $y\in \children(x)$, $G[T_y]$ does not admit a robust MIS that excludes $v(y)$. As the labeling of $y$ is correct, $y$ is not labeled \typePO. In both cases, $v(x)$ cannot be labeled \typePO by the rule.
Finally, if $G_a[T_x]$ does not admit a robust MIS that excludes $v(x)$, then because $v(x)=v(y)$, there exists at least one $y\in \children(x)$ such that $G_a[T_y]$ does not admit a robust MIS that excludes $v(y)$. As the labeling of $y$ is correct, $y$ is not labeled \typePE and hence $v(x)$ is not labeled \typePE by the rule.
\end{proof}

\subsubsection{Bridge nodes ($x \in \B$)}

If $x$ is a bridge node, then it has exactly one child $y$ whose attachment vertex $v(y)$ is a neighbor of $v(x)$. The rule consists of transforming the existing constraints on $v(y)$ into constraints on $v(x)$ as follows.

\begin{labrule}\label{rule:B}
If $L(y)$ contains label \typePI, then label $\typePO$ is added to $L(x)$; if $L(y)$ contains label \typePE, then \typePI is added to $L(x)$; if $L(y)$ contains label \typePO, then \typePI and \typePE are added to $L(x)$. Finally (cleaning), if both \typePE and \typePO have been added to $L(x)$, then \typePE is removed from $L(x)$ (see Remark~\ref{rem:POPE}).
\end{labrule}

The formal description of Labeling rule~\ref{rule:B} is given by Lines~15 to~23 in Algorithm~\ref{algo:labelnode}. In our example, the nodes $\{3,4\}$, $\{3,5\}$, and  $\{18,20\}$ are all tagged \typePI and \typePO. 

\begin{lemma}\label{lem:labelnodebrmis}
  If $x \in \B$ and $y$
  is correctly labeled, then Labeling rule~\ref{rule:B} produces a correct labeling of $x$.
\end{lemma}

\begin{proof}
  Let $x\in\B$ and $\children(x)=\{y\}$.
  \begin{itemize}
  \item If $L(y)$ contains label \typePI, then $G[T_y]$ admits a robust MIS including $v(y)$. Because $v(y)$ is included, the MIS is also valid in $G[T_x]=G[T_y]\cup G[\{x,y\}]$, and because this edge is a bridge (\ie it is not removable), it remains robust in $G[T_x]$, thus $x$ can be labeled \typePO.
  \item If $L(y)$ contains label \typePE, then $G_a[T_y]=G[T_x]$ admits a robust MIS including $v(x)$, thus $x$ can be labeled \typePI.
  \item If $L(y)$ contains label PO, then $G[T_y]$ admits a robust MIS that excludes $v(y)$.\\
    - As $v(y)$ is the only neighbor of $v(x)$ in $G[T_x]$, adding $v(x)$ to such an MIS would produce a valid MIS in $G[T_x]=G[T_y]\cup G[\{x,y\}]$, and because $\{x,y\}$ is a bridge, the resulting MIS would remain robust. Thus $x$ can be labeled \typePI.\\
    - As $v(y)$ already has a neighbor included in the MIS in $G[T_y]$, not including $v(x)$ in the MIS means that $v(x)$ is the only uncovered vertex in $G[T_x]$, which can be remedied by including to the MIS an aerial neighbor in $G_a[T_x]$. Thus, $x$ can be labeled \typePE.
  \end{itemize}

We focus now on the negative direction. 
If $G[T_x]$ does not admit a robust MIS that includes $v(x)$, then because $\{x,y\}$ is not a removable edge, $G[T_y]$ does not admit a robust MIS that excludes $v(y)$ and $G_a[T_y]=G[T_x]$ does not admit a robust MIS that excludes $v(y)$ (and includes $a$). As the labeling of $y$ is correct, $y$ is not labeled \typePI or \typePE and hence $v(x)$ is not labeled \typePI by the rule.
Similarly, if $G[T_x]$ does not admit a robust MIS that excludes $v(x)$, then because $\{x,y\}$ is not a removable edge, $G[T_y]$ does not admit a robust MIS that includes $v(y)$. As the labeling of $y$ is correct, $y$ is not labeled \typePI and hence $v(x)$ cannot be labeled \typePO.
Finally, if $G_a[T_x]$ does not admit a robust MIS that excludes $v(x)$, then because $\{x,y\}$ is not a removable edge, $G[T_y]$ does not admit a robust MIS that excludes $v(y)$. As the labeling of $y$ is correct, $y$ is not labeled \typePO and hence $v(x)$ is not labeled \typePE by the rule.
\end{proof}

\begin{algorithm}
\caption{\textbf{LabelNode}$(x)$}\label{algo:labelnode}
\begin{footnotesize}
\textbf{Parameters:} A node $x$ whose children in $T$ are already labeled\\
\textbf{Return:} A correct labeling of $x$
\begin{tabbing}
XX: \= XX \= XX \= XX \= XX \= XX \= XX \= XX \=\kill
01: \> $L(x) \gets \emptyset$\\
02:\\
03: \> \textbf{case} $x \in \PV$:\label{line:P-begin}\\
04: \> \> $L(x) \gets \{\typePI,\typePE\}$\label{line:P-end}\\
05:\\
06: \> \textbf{case} $x \in \A$:\\
07: \> \> \textbf{if} $\forall c\in \children(x), \typePI \in L(c)$ \textbf{then}\\
08: \> \> \> $L(x)\gets L(x)\cup\typePI$\\

09: \> \> \textbf{if} $\forall c\in \children(x), \typePO \in L(c)$ or $\typePE \in L(c)$ \textbf{then}\\
10: \> \> \> \textbf{if} $\exists c\in \children(x), \typePO \in L(c)$ \textbf{then}\\
11: \> \> \> \> $L(x)\gets L(x)\cup\typePO$\\
12: \> \> \> \textbf{else}\\
13: \> \> \> \> $L(x)\gets L(x)\cup\typePE$\\

14:\\
15: \> \textbf{case} $x \in \B$: (with child side $c$)\\
16: \> \> \textbf{if} $\typePI \in L(c)$ \textbf{then}\\ 
17: \> \> \> $L(x)\gets L(x)\cup \typePO$\\
18: \> \> \textbf{if} $\typePE \in L(c)$ \textbf{then}\\
19: \> \> \> $L(x)\gets L(x)\cup \typePI$\\
20: \> \> \textbf{if} $\typePO \in L(c)$ \textbf{then}\\
21: \> \> \> $L(x)\gets L(x)\cup \{\typePI,\typePE\}$\\
22: \> \> \textbf{if} $\typePO \in L(x)$ and $\typePE \in L(x)$ \textbf{then}\\
23: \> \> \> $L(x)\gets L(x)\setminus \typePE$\\

24:\\
25: \> \textbf{case} $x \in \C$:\\
26: \> \> \textbf{if} \textbf{isSatisfiable}$(x,{\tt in})$ \textbf{then}\\
27: \> \> \> $L(x)\gets L(x)\cup \typePI$\\
28: \> \> \textbf{if} \textbf{isSatisfiable}$(x,{\tt out})$ \textbf{then}\\
29: \> \> \> $L(x)\gets L(x)\cup \typePO$\\
30: \> \> \textbf{else}\\
31: \> \> \> $L(\parent(x))\gets \typePO$\\
32: \> \> \> \textbf{if} \textbf{isSatisfiable}$(x,{\tt out})$ \textbf{then}\\
33: \> \> \> \> $L(x)\gets L(x)\cup \typePE$\\
34: \> \> \> $L(\parent(x))\gets \emptyset$\\
35:\\
36: \> {\bf return} $L(x)$
\end{tabbing}\vspace{-8pt}
\end{footnotesize}
\end{algorithm}

\subsubsection{Biconnected components ($x \in \C$)}
\label{sec:biconnected}

As discussed in Section~\ref{sec:biconnectivity}, when $G$ itself is a biconnected component, an RMIS exists if and only if $G$ is bipartite (Lemmas~\ref{obs:bipartiteness} and~\ref{lem:nonbi}). When a biconnected component $x$ is a node in the \ABCT $T$, the situation is more complex due to the existence of constraints from neighbors in $T$. In particular, an RMIS satisfying these constraints may or may not exist. 

Let $G_x$ be the subgraph of $G$ induced by node $x$ (in fact, coinciding with $x$).
By construction of the \ABCT, every neighbor of $x$ in $T$ (whether children or parent) corresponds to an articulation point $y$ in $\A$ such that $v(y) \in G_x$. Thus, potential constraints on these nodes can be seen as applying to vertices {\em inside} $G_x$. This nice feature allows us to focus on finding an RMIS in $G_x$ without worrying about the entire subtree $G[T_x]$. Indeed, if an RMIS $M$ satisfying the constraints exists in $G_x$ and if every child $y$ is correctly labelled, then an RMIS $M_y$ must exist in every $G[T_y]$ such that $v(y)$ has the same status (included or excluded) in $M_y$ and in $M$. As every $y$ is an articulation point, the union of all these RMISs forms a robust MIS in $G[T_x]$.

Based on the above observation, we now focus on the simpler problem of finding an RMIS in $G_x$ that satisfies potential constraints on its articulation points. In addition, we add an input parameter to the problem that forces the status of the attachment vertex $v(x)$ in the RMIS, so that (for instance) $x$ receives label \typePI if an RMIS is found that includes $v(x)$. (The exact labeling rule using will be described later on.) 

\paragraph{\ref{sec:biconnected}.1 ~Solving $G_x$}~\smallskip

Here, we consider the following problem. Given a biconnected component $G_x$, some vertices of which are articulation points in $G$ (corresponding to nodes in $\A$); given constraints on these articulations points, encoded into the labeling function $L$; given a parameter ${\tt in, out,}$ or ${\tt none}$; the goal is to determine if there exists a robust MIS in $G_x$ in which $v(x)$ is included in the RMIS if the parameter is ${\tt in}$, not included if the parameter is ${\tt out}$, and unrestricted if the parameter is ${\tt none}$, and such that the inclusion status of every constrained articulation point $y$ matches at least one of its labels. Namely, if $L(y) \ne \emptyset$ (\ie $y$ is constrained), then $v(y)$ can be included to the RMIS only if $L(y)$ contains \typePI; and it can be excluded from the RMIS only if $L(y)$ contains \typePO or \typePE.


\paragraph{Bipartiteness {\it vs.} \typePO labels}
As already explained, if a biconnected graph is not bipartite, then weak vertices must exist (see Lemma~\ref{lem:nonbi} and Definition~\ref{def:weak}). Bipartiteness still plays a central role in the case of biconnected components, but the existence of \typePO labels makes it more subtle, as explained in the following remark.

\begin{remark}
  \label{rem:PO}
Let $y$ be an articulation point such that $v(y) \in G_x$ and $y$ contains label \typePO, then there exists an RMIS in $G[T_y]$ that does not include $v(y)$ itself but includes one (or possibly several) neighbors in $G[T_y]$ that prevent $v(y)$ from being a weak vertex.
\end{remark}

As a result, the vertices admitting a label \typePO are not subject to the same bipartiteness constraint as the others vertices.

\paragraph{The procedure} The procedure is formally described in Algorithm~\ref{algo:testrmis}, to which the reader is referred for details, and its correctness is proved subsequently. Let $E^\typePO$ be the set of edges in $G_x$ such that both endpoints contain label \typePO, and let $X$ be the subgraph $G_x \setminus E^\typePO$, possibly resulting into several connected components $X_1, ..., X_k$.
If $X$ is not bipartite, then the algorithm rejects, because this means that at least one non-PO vertex must be {\em weak}.  Otherwise, the component may possibly admit an RMIS if the combination of constraints induced by all the vertices is satisfiable. The rest of the procedure consists of encoding these constraints into a 2-SAT formula such that $G_x$ admits an RMIS (with given status for $v(x)$, the attachment vertex), if and only if the formula is satisfiable. The formula is built as follows.
For each connected component $X_i$, a SAT variable $x_i$ is created. As $X_i$ is bipartite, every vertex $v$ of one part receives label $\ell(v) \gets x_i$ and every vertex $v$ of the other part receives a label $\ell(v) \gets \neg x_i$. (Intuitively, if the eventual formula is satisfiable with $x_i=true$, then the vertices in the first part are included in the RMIS; if it is satisfiable with $x_i=false$, then the vertices of the second part are included.) Then, the existing constraints of articulation points are incorporated; namely, if an articulation point $y \in \A$ contains {\em only} label \typePI, then $\ell(v(y))$ is set to {\tt true}; if it contains only \typePO or only \typePE, then $\ell(v(y))$ is set to {\tt false}; if it contains both options (\ie \typePI\typePO or \typePI\typePE), then no constraint is added. Finally, although the edges of $E^\typePO$ induce no constraints for bipartiteness, we must still make sure that their endpoints are not both included to the RMIS (for independence), thus if $\{u,v\} \in E^\typePO$, then the clause $\neg\ell(u)\vee\neg\ell(v)$ is added.

 


\begin{algorithm}
\caption{\textbf{isSatisfiable}$(x,flag)$}\label{algo:testrmis}
\begin{footnotesize}
\textbf{Parameters:} A node $x\in \C$ and a $flag \in\{{\tt in,out,none}\}$ that forces the status of $v(x)$\\
\textbf{Return:} ${\tt true}$ if a suitable set of $x$ exists, ${\tt false}$ otherwise
\begin{tabbing}
XX: \= XX \= XX \= XX \= XX \= XX \= XX \= XX \=\kill
01: \> Let $E^{\tt PO}$ be the set of edges $\{u,v\}$ of $G_x$ such that $\typePO \in L(u)$ and $\typePO \in L(v)$\\
02: \> Let $X=G_x \setminus E^{\tt PO}$\\
03: \> \textbf{if} $X$ is not bipartite \textbf{then}\\
04: \> \> {\bf return} false\\
05: \> Let $X_1,\ldots,X_k$ be the (bipartite) connected components in $X$\\
06: \> Let $F$ be an empty $2$-SAT expression on a set $\{x_1,\ldots,x_k\}$ of boolean variables\\
07: \> \textbf{foreach} $i$ in $1..k$ \textbf{do}\\
08: \> \> Assign to each vertex $v$ in $X_i$ a label $\ell(v)\in\{x_i,\neg x_i\}$ such that neighbors have different labels\\
09: \> \textbf{foreach} $a\in \A | v(a) \in G_x$ \textbf{do}\\
10: \> \> \textbf{if} $L(a) = \{\typePI\}$ \textbf{then}\\
11: \> \> \> $F\gets F\wedge(\ell(v(a)))$\\
12: \> \> \textbf{if} $L(a) = \{\typePO\}$ or $L(a) = \{\typePE\}$ \textbf{then}\\
13: \> \> \> $F\gets F\wedge(\neg\ell(v(a)))$\\
14: \> \textbf{foreach} $\{u,v\}\in E^{\tt PO}$ \textbf{do}\\
15: \> \> $F\gets F\wedge(\neg\ell(u)\vee\neg\ell(v))$\\
16: \> \textbf{if} $flag = {\tt in}$ \textbf{then}\\
17: \> \> $F\gets F\wedge(\ell(v(x)))$\\
18: \> \textbf{else if} $flag = {\tt out}$ \textbf{then}\\
19: \> \> $F\gets F\wedge(\neg\ell(v(x)))$\\
20: \> \textbf{return} whether $F$ is satisfiable (using an external 2-SAT solver)\\
\end{tabbing}\vspace{-8pt}
\end{footnotesize}
\end{algorithm}

\paragraph{Correctness}
\newcommand{\GP}{\ensuremath{G_x^+}\xspace}

To formalize the properties of the procedure ${\tt isSatisfiable}$, we need to introduce some definitions.
Let \GP be the graph built from $G_x$ by adding a small gadget to every $v(y) \in G_x$ when $y$ is constrained (\ie $L(y)$ is non-empty). The gadget consists of a path of length $2$ incident to $v(y)$ through virtual vertices $v'(y)$ and $v''(y)$ (see Figure~\ref{fig:GP} for an intuition). This construction is not used in the procedure itself, only in the proof.
\begin{figure}[h]
  \centering
  \begin{tikzpicture}
    \tikzstyle{every node}=[defnode]
    \path (1,2) node (v2){};
    \path (2,2) node (v6){};
    \path (3,2) node (v8){};
    \path (0.3,1) node (v1){};
    \path (3.7,1) node (v9){};
    \path (1,0) node (v13){};
    \path (2,0) node (v11){};
    \path (3,0) node (v10){};
    
    \draw (v2)+(0,.5)--(v2);
    \draw (v6)+(0,.5)--(v6);
    \draw (v8)+(0,.5)--(v8);
    \draw (v8)+(1,0)--(v8);
    \draw (v11)+(0,-.5)--(v11);
    \draw (v1)--(v2)--(v6)--(v8)--(v9)--(v10)--(v11)--(v13)--(v1);
    \draw (v2)--(v11);

    \tikzstyle{every node}=[font=\scriptsize, inner sep=2pt]
    \path (v1) node[above left] {1};
    \path (v2) node[above right] {2};
    \path (v6) node[above right] {6};
    \path (v8) node[above left] {8};
    \path (v9) node[left=2pt] {9};
    \path (v10) node[below=-1pt, xshift=8pt] {$10 =$\small\,$v(x)$};
    \path (v11) node[below right] {11};
    \path (v13) node[below right] {13};

    \tikzstyle{every node}=[font=\footnotesize, inner sep=2pt]
    \path (v2) node[left] {\typePI\hspace{-1.5pt}\typePO};
    \path (v6) node[below] {\typePI\hspace{-1.5pt}\typePO};
    \path (v11) node[above=1pt,xshift=9pt] {\typePI\hspace{-1.5pt}\typePO};
    \path (v8) node[below left,xshift=2pt] {\typePI};
    \draw[dashed,thick,gray,rounded corners] (0,2.5) rectangle (4,-.5);
    \path (0,2.2) node[left] {\large $x=$ {\tt B}};
    \path (5,1) node {\LARGE $\rightarrow$};
    \path (0,-.85) coordinate (bidon);
  \end{tikzpicture}
  \hspace{10pt}
  \begin{tikzpicture}
    \tikzstyle{every node}=[defnode]
    \path (1,2) node (v2){};
    \path (1,2.4) node[gray] (v2b){};
    \path (1,2.8) node[gray] (v2c){};
    \path (2,2) node (v6){};
    \path (2,2.4) node[gray] (v6b){};
    \path (2,2.8) node[gray] (v6c){};
    \path (3,2) node (v8){};
    \path (3,2.4) node[gray] (v8b){};
    \path (3,2.8) node[gray] (v8c){};
    \path (0.3,1) node (v1){};
    \path (3.7,1) node (v9){};
    \path (1,0) node (v13){};
    \path (2,0) node (v11){};
    \path (2,-.4) node[gray] (v11b){};
    \path (2,-.8) node[gray] (v11c){};
    \path (3,0) node (v10){};
    
    \tikzstyle{every node}=[circle, draw, red, inner sep=2pt]
    \path (v8) node {};
    \path (v10) node {};
    \path (v13) node {};
    \path (v2) node {};
    \path (v2c) node {};
    \path (v6b) node {};
    \path (v8c) node {};
    \path (v11b) node {};

    \draw (v2)--(v2b)--(v2c);
    \draw (v6)--(v6b)--(v6c);
    \draw (v8)--(v8b)--(v8c);
    \draw (v11)--(v11b)--(v11c);
    \draw (v1)--(v2)--(v6)--(v8)--(v9)--(v10)--(v11)--(v13)--(v1);
    \draw (v2)--(v11);

    \tikzstyle{every node}=[font=\scriptsize, inner sep=2pt]
    \path (v1) node[above left] {1};
    \path (v2) node[above right] {2};
    \path (v6) node[above right] {6};
    \path (v8) node[above left] {8};
    \path (v9) node[left=2pt] {9};
    \path (v10) node[below=-1pt, xshift=8pt] {$10 =$ \normalsize $v(x)$};
    \path (v11) node[below right] {11};
    \path (v13) node[below right] {13};

    \tikzstyle{every node}=[font=\footnotesize, inner sep=2pt]
  \end{tikzpicture}
  \hspace{10pt}
  \caption{\label{fig:GP}Example of component $G_x$ (left)--$x$ corresponds to the component {\tt B} in
  Figure~\ref{fig:treecomp}. With $v(x)=10$, a chain of two virtual nodes is added to each constrained node $v(y) \in \{2,6,8,11\}$
  that gives the corresponding $G_x^+$ (right). The MIS shown on $G_x^+$ corresponds to the set $S$ used in the proof 
  of Lemma~\ref{lem:testrmis1} with $v(x)$ included.}
\end{figure}
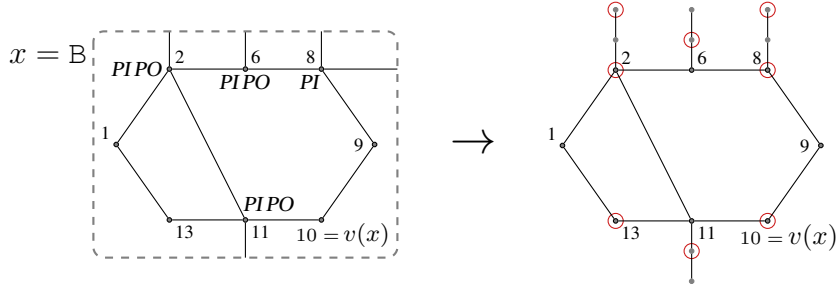
We say that a set of vertices $S$ is a {\em suitable RMIS} of \GP if and only if $S$ is an RMIS of \GP such that, for every $y\in G_x$ with at least one label in $L(y)$, we have:\\
- $v(y)$ in $S$ only if \typePI $\in L(y)$;\\
- $v(y)$ not in $S$ only if \typePO or \typePE appears in $L(y)$; and\\
- $v(y)$ not in $S$, $v'(y)$ in $S$, and $v''(y)$ not in $S$ only if PO appears in $L(y)$.

The idea is that, for every $y$, the graph induced by the path $\{v(y),v'(y),v''(y)\}$ replaces $G[T_y]$ in the management of constraints so that a suitable RMIS of \GP exists if and only if an RMIS of $G[T_x]$ exists. We now prove the main lemma.

\begin{lemma}\label{lem:testrmis1}
Given a biconnected component $x\in\C$ with constraints on some articulation points, ${\tt isSatisfiable}(x,{\tt in})$ returns true if and only if a suitable RMIS of $\GP$ including $v(x)$ exist; ${\tt isSatisfiable}(x,{\tt out})$ returns true if and only if a suitable RMIS of $\GP$ excluding $v(x)$ exist;  ${\tt isSatisfiable}(x,{\tt none})$ returns true if and only if a suitable RMIS of $\GP$ exists, irrespective of $v(x)$.
\end{lemma}

\begin{proof}
Let $x$ be the considered node in the \ABCT and $G_x$ the corresponding component. Let $\GP$ be the augmented version of $G_x$ as described above. First observe that, if $G_x \setminus E^\typePO$ is bipartite, then so is $\GP \setminus E^\typePO$, because the adjunction of such paths cannot affect bipartiteness.
In the following, we first prove the direction that (1) if the procedure returns true, then the corresponding type of suitable RMIS exists; then (2) if the given type of suitable RMIS exists, then the procedure with corresponding parameters must return true.

\begin{enumerate}[leftmargin=15pt]
\item[(1)] If $F$ is satisfiable, define $S$ as the set of vertices induced by a positive assignment, \ie $\{v \in G_x | \ell(v)=true\}$. Now, add to $S$ some of the virtual vertices as follows: if $v(y)$ is in $S$, add only $v''(y)$; if $v(y)$ is not in $S$ and is weak in $G_x$, add only $v'(y)$; if $v(y)$ is not in $S$ and is not weak in $G_x$, add only $v''(y)$. Lines 16 to 19 ensure that $v(x)$ is included in the MIS if $flag={\tt in}$ and that it is excluded if $flag={\tt out}$ (if $flag={\tt none}$, no constraint applies). We now prove that $S$ is independent, maximal, robust, and suitable in $G^+$.
  \smallskip

\noindent{\em Independence:} Let $e=\{u,v\}$ be an edge of \GP. If $e$ is an edge from one of the extra paths, $u$ and $v$ cannot be both in $S$ by construction. Otherwise (\ie $e$ is in $G_x$), either $e$ belongs to $E^{\typePO}$, or it does not. If it does, then the clause introduced on Line 15 ensures that $u$ and $v$ cannot be both in the MIS. If it does not, then it belongs to $\GP\setminus E^{\typePO}$, which is bipartite, then Line 08 (and the greedy process for extra paths) ensures that $u$ and $v$ cannot be both in the MIS.
  \smallskip

\noindent{\em Maximality and robustness:} 
Recall that independence is not affected by the removal of edges (Observation~\ref{obs:independence}), thus robustness means preserving maximality. Let $u$ be a vertex of $G_x^+$ that does not belong to $S$. Either $u$ belongs to an edge in $E^{\typePO}$ or it does not. If it does, then by Remark~\ref{rem:PO} it has a neighbor in the MIS (maximality) and it cannot be a weak vertex (robustness).
If it does not, then {\em all} of its neighbors in $G_x^+$ are in the set because $G_x \setminus E^{\typePO}$ is bipartite (and the extension of $S$ in the extra paths keep alternating the inclusion status), thus so long as $G_x^+$ remains connected, $u$ cannot be added to $S$ (maximality and robustness). 
  \smallskip

\noindent{\em Suitability:} 
For every $v(y)\in G_x$ such that $L(y) \ne \emptyset$. Lines 12 and 13 ensure that $v(y)\in S$ only if \typePI appears in $L(y)$. Lines 10 and 11 ensure that $v(y)\notin S$ only if \typePO or \typePE appear in $L(y)$.
If $v'(y) \in S$ and $v''(y) \notin S$, then by construction $v(y) \notin S$ and it is weak in $G_x$. Line 08 implies that at least one adjacent edge to $v(y)$ in $G_x$ is in $E^\typePO$ and then \typePE does not appear in $L(y)$, implying that $PO$ does.
  \smallskip


\item[(2)] Let $M$ be a suitable RMIS of $\GP$ in which the status of $v(x)$ is ${\tt in}$ (resp ${\tt out}$). Let $X = G_x\setminus E^{\typePO}$. If $X$ is not bipartite, then it contains a weak vertex, which contradicts the robustness of $M$. Thus, $X$ must be bipartite. In general, $X$ may be made of several connected components $X_1, ..., X_k$.
A satisfying assignment can then be constructed from $M$ as follows. For all $i$ ranging from $1$ to $k$, chose the value of $x_i$ so that $\forall v\in G_x, \ell(v)=true\Leftrightarrow v\in M$. This assignment satisfies the mutual exclusivity constraint in Line 15 (with respect to edges in $E^\typePO$), as otherwise $M$ would not be independent, and it also satisfies the constraint added with respect to $L(y)$ (Lines 09 to 13), by suitability of $M$. Finally, due to Lines 17 and 19, the assignment must include (respectively exclude) $v(x)$ if the input flag is ${\tt in}$ (resp. ${\tt out}$). Thus a satisfiable assignment corresponding to $M$ must exists.\vspace*{-20pt}
\end{enumerate}
\end{proof}
%


\paragraph{\ref{sec:biconnected}.2 ~The Labeling Rule}~\smallskip

Let us now return to the labeling of node $x \in \C$ in the \ABCT, assuming that all $y$ in $\children(x)$ are correctly labeled. The goal here is to label $x$ as per the possibilities, namely to assign label \typePI if an RMIS exists in $G[T_x]$ that includes $v(x)$; to assign label \typePO if an RMIS exists in $G[T_x]$ that excludes $v(x)$; and, in case the latter does not, to assign \typePE if an RMIS exists that excludes $v(x)$ provided that an external neighbor of it is subsequently included.

We need to distinguish here between the case that $x$ is the root component $R$ of the \ABCT, or just an internal node. If $x$ is the root, then it has no attachment point and the only thing that matters is whether an RMIS exists or not. As such, $x$ does not receive a labeling, and is instead treated directly from the main algorithm (described in the next section). Thus, we focus on how the above procedure can be used to label $x$ when $x$ is an internal node of the \ABCT.

The first two tests are realized by fixing the status of $v(x)$ as a parameter when calling ${\tt isSatisfiable()}$. A first call is made, setting this parameter to {\tt in}, then a second call is made, setting this parameter to ${\tt out}$. If the second call fails (no such RMIS exist), then a third call is made with a different strategy. This strategy relies on a trick using label \typePO. Remark~\ref{rem:PO} implies that whenever a child $y$ of $x$ is labeled \typePO, then $v(y)$ (itself in $x$) does not need to have a neighbor {\em within} $x$ that is included in the RMIS. As a result, testing whether $G[T_x]$ admits an RMIS if an external neighbor of $v(x)$ is subsequently added (\ie label \typePE) can be done by pretending that the parent $z$ of $x$, whose attachment vertex $v(z)=v(x)$ is in $x$, itself has an external neighbor in the MIS. Thus, just like the other children labeled \typePO, an artificial \typePO label is assigned to $\parent(x)$, resulting in the forced non-selection of $v(x)$ to the MIS and the adjunction of an extra path to $v(x)$ (whose neighbor in the path {\em will} be included). The existence of an RMIS satisfying this configuration implies that $x$ can be labeled \typePE. The resulting labeling rule is as follows.

\begin{labrule}\label{rule:C}
  If ${\tt isSatisfiable()}$ returns true with parameter {\tt in}, then \typePI is added to $L(x)$. Next, if ${\tt isSatisfiable()}$ returns true with parameter {\tt out}, then \typePO is added to $L(x)$. Next, if $\typePO\notin L(x)$ and ${\tt isSatisfiable()}$ returns true with parameter {\tt out} with an artificial \typePO label in $L(parent(x))$, then \typePE is added to $L(x)$.
\end{labrule}

The corresponding instructions are formally given in Lines~25 to~34 of Algorithm~\ref{algo:labelnode}.
In our example (Figure~\ref{fig:ex_compo}), component $\texttt{E}$ is labeled both \typePI (with the suitable set $\{28,30\}$) and \typePO (with the suitable set $\{29,31\}$) while \texttt{B} is labeled \typePI (with the suitable set $\{2,8,10,13\}$) and \typePE (with the suitable set $\{1,8,11\}$).

\begin{lemma}\label{lem:labelnodebrmis}
  If $x \in \C\setminus \{R\}$ and all nodes of $children(x)$ are correctly labeled, then Labeling rule~\ref{rule:C} produces a correct labeling of $x$.
\end{lemma}

\begin{proof}
  Follows from Lemma~\ref{lem:testrmis1} and definition of a suitable RMIS of $\GP$.
\end{proof}



\subsection{The actual algorithm}
\label{sec:actual-algorithm}

Now that the \ABCT decomposition and the labeling rules are described, the main algorithm is relatively easy to present. We start by giving an informal description of the algorithm,  along with the corresponding code, assuming that the \ABCT decomposition is given. Then, we give an example of execution in which the graph of Figure~\ref{fig:treecomp} and~\ref{fig:ex_compo} is entirely labeled. The time complexity of the algorithm is discussed in a subsequent section, in which we also discuss the cost of computing the \ABCT. Then, in Section~\ref{sec:constructivity}, we discuss how to turn this decision algorithm into a constructive algorithm that returns an RMIS (if one exists).

\subsubsection{Description}
Given the \ABCT $T$, the algorithm starts by chosing an arbitrary biconnected component $\R\in \C$ to be the root node of $T$ and setting it as the current node. (As already explained, if $G$ has no biconnected component, then it is itself a tree and any property is trivially robust.) The main component of the algorithm consists of a DFS recursion within $T$ starting from $R$. For every node $x$ in $T$, the actual labeling of $x$ occurs after the last child of $x$ has been visited by the DFS, \ie after the labels of all the children are known. At any point of the execution, if the set of labels of a node is empty {\em after} executing the labeling rules, this means that the current subtree (and a fortiori $G$ itself) does not admit a robust MIS, thus the algorithm rejects (and terminates). If the execution survives until the recursion returns at the root $R$, then a special call to ${\tt isSatisfiable}$ is performed without constraints on the attachment vertex (which does not exists at $R$). If this call returns true, then $G$ admits an RMIS, otherwise it does not. The code of the main algorithm is given on Algorithm~\ref{algo:findrmis}.

\begin{algorithm}[h]
\caption{Main algorithm}\label{algo:findrmis}
\begin{footnotesize}
\textbf{Input:} An \ABCT $T$ whose nodes are $\A\cup\B\cup\C\cup\PV$, rooted in some $R\in \C$\\
\textbf{Output:} A labeling $L$ of $T$
\begin{tabbing}
XX: \= XX \= XX \= XX \= XX \= XX \= XX \= XX \=\kill
01: \> \textbf{foreach} $x\in\ \children(R)$ \textbf{do}\\
02: \> \> \textbf{LabelSubTree}$(x)$\\
03: \> \textbf{if} \textbf{isSatisfiable}$(R,{\tt none})$ \textbf{then}\\
04: \> \> {\bf accept}\\
05: \> \textbf{else}\\
04: \> \> {\bf reject}

\end{tabbing}\vspace{-8pt}
\end{footnotesize}
\end{algorithm}

The recursion of the DFS is given by the procedure in Algorithm~\ref{algo:labelsubtreermis}, which consists of descending the children first, then calling the labeling procedure described on Algorithm~\vref{algo:labelnode}, and finally reject if none of the labels were possible for the current node.

\begin{algorithm}[h]
\caption{Function \textbf{LabelSubTree}$(x)$}\label{algo:labelsubtreermis}
\begin{footnotesize}
\textbf{Parameters:} A node $x \in V_T$
\begin{tabbing}
XX: \= XX \= XX \= XX \= XX \= XX \= XX \= XX \=\kill
01: \> \textbf{foreach} $c\in \children(x)$ \textbf{do}\\
02: \> \> \textbf{LabelSubTree}$(c)$\\
03: \> $L(x) \gets$ \textbf{LabelNode}$(x)$\\
04: \> \textbf{if} $L(x)=\emptyset$ \textbf{then}\\
05: \> \> {\bf reject}
\end{tabbing}\vspace{-8pt}
\end{footnotesize}
\end{algorithm}

The correctness of the whole mainly follows from the correctness of the relabeling rules. Indeed, the notion of {\em correct labeling} (see Definition~\ref{def:correct-labeling}) specifies that the label of a node encodes the existence of RMIS in the current subtree (parametrized by the status of the highest vertex). On the other hand, Lemmas \ref{lem:labelnodeprmis} to \ref{lem:labelnodebrmis} guarantee us that if all the children of a node $x$ are correctly labeled, then the appropriate labeling rule will produce a correct labeling of $x$ (unless the algorithm rejects). If thus follows, by induction on the tree defined by the DFS, that if the execution survives the recursion started at Lines 01 and 02 of Algorithm~\ref{algo:findrmis}, then all the children of the root node $R$ are correctly labeled. And Lemma~\ref{lem:testrmis1} allows us to conclude based the call to ${\tt isSatisfiable}$ with parameter ${\tt none}$.

\subsubsection{Example}

The outcome of the labeling is shown in Figure~\ref{fig:tagged}, corresponding to the input graph of Figure~\ref{fig:treecomp} and its \ABCT shown in Figure~\ref{fig:ex_compo}.

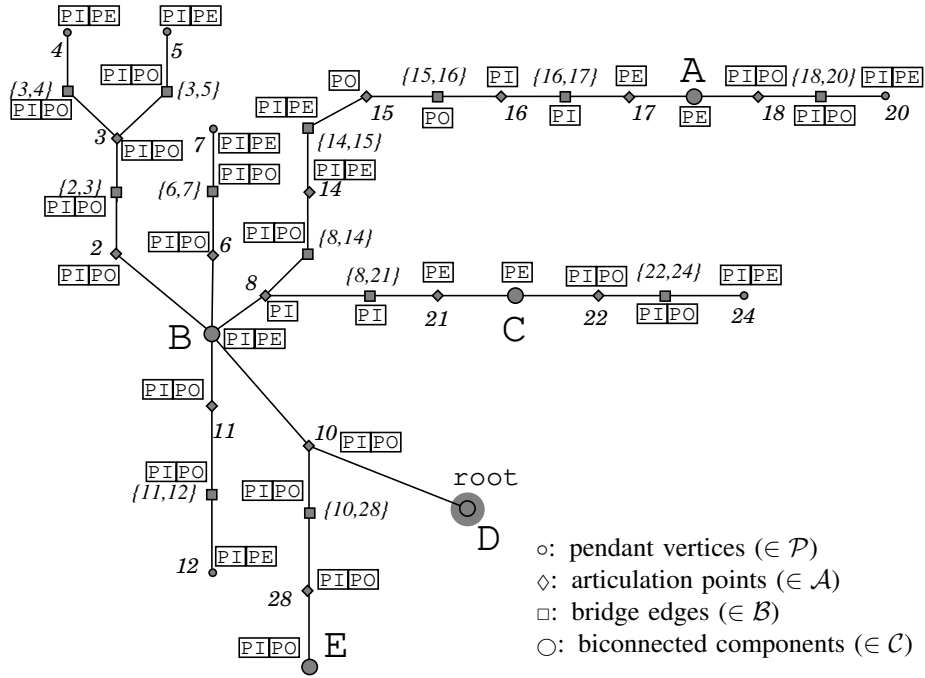
\begin{figure}[h]
\centering
\includegraphics[width=1\textwidth] {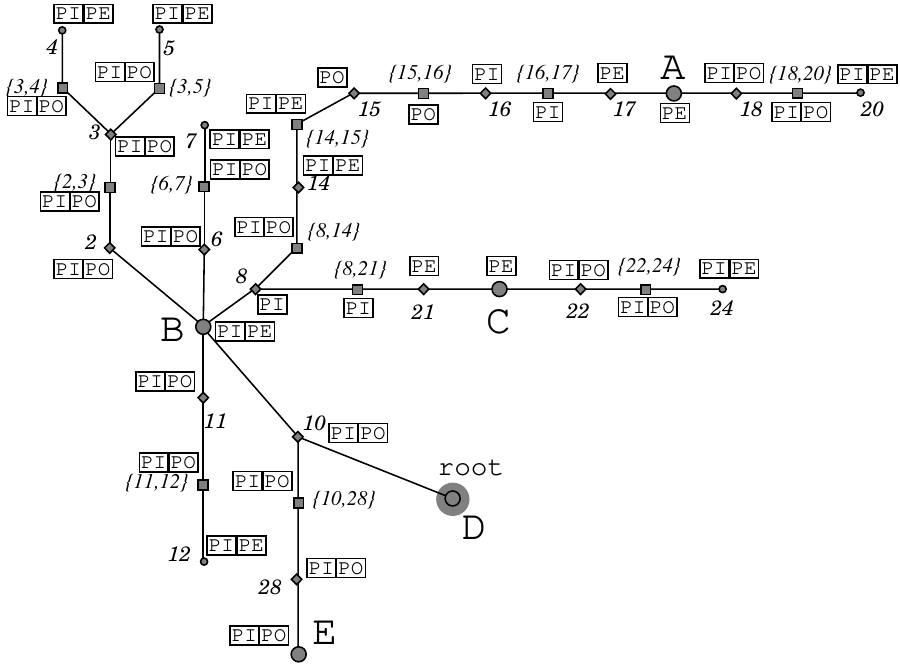}
  \vspace{-2.4cm}
  
  \hfill
  \begin{tikzpicture}
    \path (0,0) node[text width=5cm](legend){
      {\small $\circ$}: pendant vertices ($\in \PV$) \\
      {\scriptsize $\Diamond$}: articulation points ($\in \A$)\\
      {\tiny $\Box$}: bridge edges ($\in \B$)\\
      {\scriptsize $\bigcirc$}: biconnected components ($\in \C$)
    };
  \end{tikzpicture}
\caption{\label{fig:tagged}Outcome of the algorithm corresponding to the input graph of Figure~\ref{fig:treecomp}.}
\end{figure}

\subsection{Time complexity}
\label{sec:complexity}

In this paragraph, we analyse the running time of the algorithm. In the following, $n$ denotes the number of vertices in $G$, which is related to the number of vertices in $T$ by a constant factor (thus the same $O$ notation is used for both interchangeably). The purpose is not to present a fine-grained analysis. Instead, we favor simple arguments over optimal ones, the resulting running time being in any case dominated by $O(n^3)$ time steps. 

\begin{lemma}
  The \ABCT decomposition can be computed in $O(n^2)$ time steps.
\end{lemma}

\begin{proof}
  The biconnected components of $G$ can be computed using a classical algorithm by Hopcroft and Tarjan~\cite{HT73}, which has the same complexity as a DFS in $G$, thus in time $O(n^2)$ in the worst case that $G$ is dense. After this process, the components of size $1$ and degree $1$ correspond to the nodes in $\PV$, the edges between components correspond to bridges is $\B$, and the endpoints of these bridges correspond to articulation points in $\A$. All of these can be determined either through looping over the components themselves (at most $O(n)$), or looping over the edges of $G$ (at most $O(n^2)$) in order to compare the membership of their endpoints to the components (this latter test takes constant time after Hopcroft and Tarjan's algorithm). Once the four types of nodes are identified, the relations between nodes in $\A$ and nodes in $\C$ can be found by looping over all the vertices of the components ($O(n)$ vertices overall, with constant time test for each) and the relations between nodes in $\B$ on the one hand, and nodes in $\PV$ or $\A$ in the other hand, can be found by looping over all the edges of $G$ ($O(n^2)$ edges overall, with constant time tests for each).
\end{proof}

\begin{lemma}
  \label{lem:2-SAT}
  Procedure ${\tt isSatisfiable()}$ takes $O(n^2)$ time steps.
\end{lemma}

\begin{proof}
  Building the set $E^{\typePO}$ can be done by looping over all the edges of $G_x$ and looking at the labels of its endpoints (in constant time), thus in $O(n^2)$ time steps. Building $X$ without these edges is also linear in the number of edges, thus costs $O(n^2)$ time steps. Then, testing if a graph is bipartite can be done by performing a 2-coloring of it along a BFS as follows: give color 1 to an arbitrary first vertex $v$, then as the BFS progresses, give color 2 to all the vertices at distance $1$ from $v$, then color 1 again to all the vertices at distance $2$ from $v$, etc. Possibly restart the BFS algorithm if $X$ is not connected. As the BFS algorithm runs in time linear in the number of edges (whether or not $X$ is connected), the $2$-coloring algorithm will take up to $O(n^2)$ time steps, and the detection of whether the coloring is proper can also be done in $O(n^2)$ time steps (checking the two endpoints of every edge). If the procedure is still running ($X$ is bipartite), the obtained $2$-coloring can be reused directly over lines 07 and 08 to assign to every vertex its SAT constraint $x_i$ or $\neg x_i$, depending on its color and connected component (the latter being possibly determined during the BFS). Then, Lines 09 to 13 consist of a loop over (part of) the vertices with constant time processing time for each. Lines 14 and 15 consist of a loop over the edges of $E^{\typePO}$ and constant processing time for each, thus at most $O(n^2)$. Lines 16 to 19 take constant time. Finally, the resulting formula is made of $O(n^2)$ constraints over $O(n)$ clauses, and it is known that 2-SAT is solvable in time linear in the number of constraints. In conclusion, the overall procedure takes $O(n^2)$ time steps.
\end{proof}

\begin{lemma}
  The decision algorithm presented in Section~\ref{sec:algo} takes $O(n^3)$ time steps.
\end{lemma}

\begin{proof}
  The main loop of the algorithm performs a depth first search in the \ABCT $T$, whose running time is $O(n)$. Every node $x$ is labeled once, after the DFS has finished the exploration of its subtree. At that moment, four cases apply depending on whether $x$ is of type $\A,\B,\C,$ or $\PV$. If $x$ is of type $\PV$, the labeling takes constant time (see Algorithm~\ref{algo:labelnode} Lines 03 to 04, and 17 to 25, respectively). If $x$ is of type $\B$, the labeling also takes constant time because $x$ has only one child (Lines 17 to 25). If $x$ is of type $\A$, then the labeling rule performs two loops over the children of $x$, thus the labeling takes $O(n)$ time steps (the nested existential quantifier can be tested upon the same loop). If $x$ is in $\C$, then two or three calls are made to ${\tt isSatisfiable()}$, whose running time is $O(n^2)$ by Lemma~\ref{lem:2-SAT}. As a result, the overall running time is dominated by the latter being applied to possibly $O(n)$ nodes, for a total of $O(n^3)$ time steps.
\end{proof}

\subsection{Constructivity}\label{sec:constructivity}

The algorithm presented over the previous subsections is a decision algorithm that returns yes or no, or equivalently, accepts or rejects the given instance. However, the algorithm is fairly easy to adapt in order to obtain an actual RMIS if one exists. There are essentially two options: either the RMIS is constructed in parallel of the decision algorithm, or a dedicated algorithm {\em redescends} the tree after the algorithm has completed. Here, we briefly sketch both solutions.
\begin{enumerate}
\item {\em Construction in parallel.} In this version, the main difficulty is that the algorithm cannot decide which of the current labels will eventually be compatible with later (higher) constraints in the tree. For example, a node $x$ may admit an RMIS in its subtree whether or not its attachment vertex $v(x)$ is included (label \typePI and \typePO), but it might later become mandatory that $v(x)$ is included. Thus, the idea is to build several RMIS simultaneously, one for each possible label of the current attachment vertex. These RMIS can be built by extending the ones of the children (with compatible status), starting from the leaves, thus at most two RMIS need to be memorized for the current subtree. In the particular case that $x$ is a biconnected component, it is needed that the satisfying assignment found by the 2-SAT solver be known and transposed into a membership status for each vertex (similarly to the set $S$ in the proof of Lemma~\ref{lem:testrmis1}).
\begin{figure}[h]
  \centering 
  \includegraphics[width=10cm] {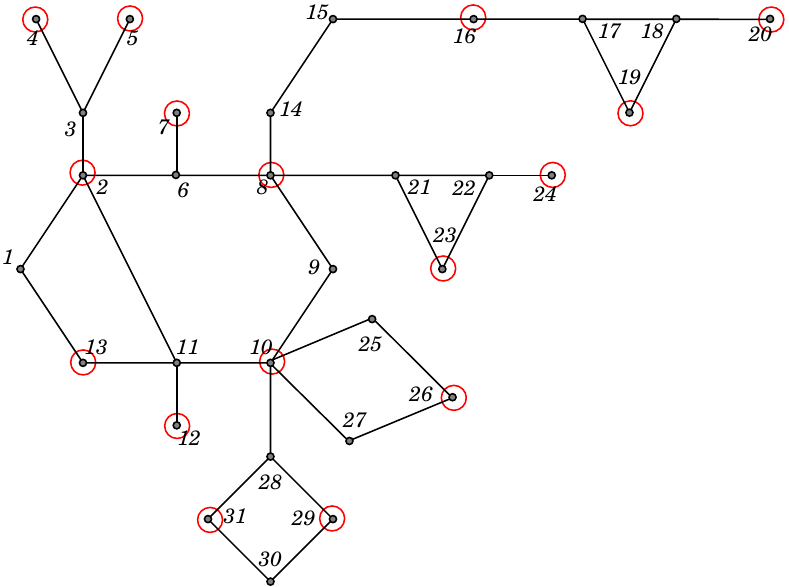}
 \caption{A possible RMIS of $G$ that includes $16$ vertices.\label{fig:ex_rmis2}}
\end{figure}
\item {\em Construction afterwards.} In this version, we are given an already labeled tree whose root component has been solved (successfully).
  The algorithm consists of redescending the labeled tree, by making {\em arbitrary choices} whenever the considered attachment vertex has several labels. Indeed, the existence of an RMIS corresponding to each such label is guaranteed when redescending the tree. The cases that $x$ is of type $\A, \B,$ or $\PV$ are trivial (just pick one of the labels of the child and translate it as a membership status for its attachment vertex). The only difficulty is when $x$ is a biconnected components. In this case, a label is chosen arbitrarily among the labels of its attachment vertex, and the 2-SAT reduction procedure is called again with the configuration corresponding to this constraint in order to obtain the satisfying assignement that specifies which inner vertex must be included to the MIS. (Alternatively, the satisfying assignments might be recorded while in the decision algorithm.) 
\end{enumerate}


A possible outcome of the constructive algorithm, obtained by manually following the first method is shown in Figure~\ref{fig:ex_rmis2}.


\section{Further discussion on the temporal interpretation of robustness}
\label{sec:temporal-interpretation}
\newcommand{\ER}{\ensuremath{{\cal E^R}}\xspace}
\newcommand{\EP}{\ensuremath{{\cal E^P}}\xspace}
\newcommand{\TCR}{{\ensuremath{{\cal TC^{\cal R}}}}\xspace}
\newcommand{\G}{\ensuremath{\mathcal{G}}\xspace}
\tikzset{circ/.style={fill=none,red,gray,thick,inner sep = 3pt}}

\tikzset{every node/.style={defnode}}
\def\carre (#1,#2){
  \path (#1,#2+1) node (a){};
  \path (#1+1,#2+1) node (b){};
  \path (#1,#2) node (c){};
  \path (#1+1,#2) node (d){};
}

As discussed in the introduction, the original motivation behind the concept of robustness comes from highly-dynamic networks. Here, we review some of the background that led to its definition. As the content pertains to the temporal interpretation, the reader interested in the notion of robustness {\it per se} can safely omit reading this section. 
In~\cite{CMM11}, three canonical ways of redefining combinatorial problems in highly-dynamic networks were explored, in particular covering problems such as {\em dominating set}, {\em vertex cover}, or {\em independent set}. The main focus in~\cite{CMM11} was on dominating sets, which are subsets of nodes $S\subseteq V$ such that each node in the network is either in $S$, or has a neighbor in $S$. 
Three natural adaptations of these problems were formulated:

\begin{itemize}
\item {\em Temporal} version: the covering property of the problem is achieved {\em over time} -- {\it E.g.,} for domination, every node outside the set must share an edge {\em at least once} over the lifetime with a node in the set.

\item {\em Evolving} version: the covering property must be satisfied at any given instant, relative to the current structural snapshot of the network; however, the solution can be updated as the structure of the network changes. This version is closer to the traditional ``dynamic graph algorithms'' (also referred to as ``reoptimization'').

\item {\em Permanent} version: the covering property must be satisfied in every snapshot (as for the evolving version), but here the solution cannot be updated.
\end{itemize}

The three versions are related. For example, in the case of dominating sets, the temporal version consists of computing a dominating set in the {\em footprint} $\cup \G$ of the network (\ie the union of all snapshots) and the permanent version consists of computing a solution in $\cap \G$ (their intersection)~\cite{CMM11}. Furthermore, solutions to the \emph{permanent} and the \emph{temporal} versions form upper and lower bounds for the \emph{evolving} version, respectively.

In~\cite{DKP15}, the temporal definition above is extended to {\em infinite} lifetime networks, by requiring that the covering relation (domination or else) be satisfied {\em infinitely often}. 
The authors observe that whenever the network is temporally connected in a recurrent way,
which corresponds to Class~{$\cal TC^R$} in~\cite{Cas18} (and Class~$5$ in~\cite{CFQS12}), one has the guarantee that among all the edges that appear at some point, at least a connected spanning subset of them must reappear forever~\cite{BDKP16}. In other words, an equivalent characterization of $\cal TC^R$ is that the {\em eventual footprint} of the network (\ie the union of those edges which reappear forever) is connected. The uncertainty as to which part of the footprint will belong to the eventual footprint is the main motivation behind robustness.



For completeness, let us cite a few recent works that considered the problem of computing temporal covering structures in highly-dynamic networks (or graphs), although not related to robustness. Mandal et al.~\cite{permanent} study approximation algorithms for the permanent version of {\em dominating sets}. Bamberger {\it et al.}~\cite{kuhn} also consider a temporal variant of vertex coloring and MIS. Finally, Akrida {\it et al.}~\cite{AMSZ18} define a variant of the temporal version in the case of vertex cover, in which a solution is not just a set of nodes (as it was in~\cite{CMM11}) but a set of pairs $(nodes, times)$, allowing different nodes to cover the edges at different times (and within a sliding time window). 

\section{Concluding remarks}
\label{sec:conclusion}

This paper introduced a new form of heredity called robustness, motivated by various kinds of dynamic networks. In particular, we believe that robustness is a key property of highly dynamic systems 
for achieving stable structures in unstable environments.  

Focusing on the classical covering problem MIS, we characterized the set of graphs in which all MISs are robust. We
gave partial characterizations of the existential analogues of this class, namely graphs that {\em admit} a robust solution. 
We characterized the class entirely by means of a polynomial time algorithm which finds a robust MIS in an arbitrary graph is one exists (and rejects otherwise).
Whether a characterization of the existential classes exists in terms of elementary graph properties is an open question. It would also be interesting to investigate the robustness of other types of structures (\eg minimal dominating sets) and of basic graph properties, with the aim to understand robustness at a deeper level. For example, bipartiteness or ``sputnikness'' are themselves robust properties of a graph.

\section*{Acknowledgment}

\noindent This research has been supported by ANR project ESTATE (ANR-16-CE25-0009-03). 

\bibliographystyle{plain} 
\bibliography{robustness}

\end{document}